\documentclass[a4paper,UKenglish,cleveref, autoref, thm-restate]{lipics-v2019}

\usepackage{amsmath, nccmath}
\usepackage{amsthm}
\usepackage{amsfonts}
\usepackage{bussproofs}
\usepackage{stmaryrd} %
\usepackage{svg}

\usepackage{xspace}
\newcommand{\hfl}[0]{HFL\texorpdfstring{$_\mathbb{N}$}{N}\xspace}
\newcommand{\buchi}[0]{B\"{u}chi\xspace}

\newif\ifdraft\draftfalse
\newif\iffull\fulltrue

\ifdraft
\newcommand\nk[1]{\textcolor{red}{[#1 -nk]}}

\newcommand\tk[1]{\textcolor{blue}{[#1 -tsukada]}}

\else
\newcommand\nk[1]{}

\newcommand\tk[1]{}

\fi
\title{A Cyclic Proof System for \hfl} %

\titlerunning{} %

\author{Mayuko Kori}{Department of Informatics, The Graduate University for Advanced Studies (SOKENDAI), Hayama, Japan \and National Institute of Informatics, Tokyo, Japan}{mkori@nii.ac.jp}{https://orcid.org/0000-0002-8495-5925}{}%

\author{Takeshi Tsukada}{Graduate School of Science, Chiba University, Chiba, Japan}{tsukada@math.s.chiba-u.ac.jp}{https://orcid.org/0000-0002-2824-8708}{}

\author{Naoki Kobayashi}{The University of Tokyo, Tokyo, Japan}{koba@is.s.u-tokyo.ac.jp}{https://orcid.org/0000-0002-0537-0604}{}

\authorrunning{M. Kori, T. Tsukada and N. Kobayashi} %

\Copyright{Mayuko Kori, Takeshi Tsukada and Naoki Kobayashi} %

\ccsdesc[100]{Theory of computation~Proof theory} %

\keywords{Cyclic proof, higher-order logic, fixed-point logic, sequent calculus} %

\category{} %

\relatedversion{} %

\supplement{}%

\nolinenumbers

\acknowledgements{We would like to thank anonymous referees for useful comments.
  This work was supported by JSPS KAKENHI Grant Number JP15H05706, JP20H00577, and JP20H05703.
  We also acknowledge Dominik Wehr for pointing out a gap in Section~\ref{sec: cyclic}.
  The first author is supported by ERATO HASUO Metamathematics for Systems Design Project (No. JPMJER1603), JST.}%

\begin{document}

\maketitle

\begin{abstract}
A cyclic proof system allows us to perform inductive
reasoning without explicit inductions.
We propose a cyclic proof system for \hfl, which is a higher-order predicate logic with natural numbers and alternating fixed-points.
Ours is the first cyclic proof system for a higher-order logic, to our knowledge.
Due to the presence of higher-order predicates and alternating fixed-points,
our cyclic proof system requires a more delicate global condition on cyclic proofs
than the original system of Brotherston and Simpson.
We prove the decidability of checking the global condition and soundness of this system,
and also prove a restricted form of standard completeness
for an infinitary variant of our cyclic proof system.
A potential application of our cyclic proof system is
semi-automated verification of higher-order programs, based on Kobayashi et al.'s recent
work on reductions from program verification to \hfl{} validity checking.
\end{abstract}

\section{Introduction}
There have recently been extensive studies on cyclic proof systems.
They allow a proof to be cyclic, as long as it
satisfies a certain sanity condition called
the ``global trace condition.''
Cyclic proofs enable inductive reasoning
without explicit inductions, which would be useful for proof automation.
Various cyclic proof systems have been proposed~\cite{brotherston, doumane,sprenger} for first-order logics, and some of them have been applied to automated program verification~\cite{brotherston11, brotherston12, tellez20}.

In the present paper, we propose a cyclic proof system for \hfl{},
a higher-order logic with natural numbers and least/greatest fixpoint
operators on higher-order predicates.
\hfl{} has been introduced by Kobayashi et al.~\cite{kobayashi19, kobayashi18}
as an extension of HFL~\cite{viswanathan04},
and shown to be useful for higher-order program verification.
Verification of various temporal properties of higher-order programs
can naturally be reduced to validity checking for \hfl{} formulas.
For example, consider the following OCaml program:
\begin{verbatim}
let rec g n = if n=0 then () else g (n-1) in
let rec f h m = h m; f h (m+1) in f g 0
\end{verbatim}
The property ``\texttt{f} is infinitely often called'' is then expressed as the following
formula:
\[
\big(\nu F.\lambda h.\lambda m. h\;m \land F\;h\;(m+1)\big)\; \big(\mu G.\lambda n.(n=0\lor (n\neq 0\land G\; (n-1)))\big)\;0.
\]
Here, \(\nu F.\lambda h.\cdots\) and \(\mu G.\lambda n.\cdots\)
respectively represent the greatest and least predicates such that
\(F=\lambda h.\cdots\) and \(G=\lambda n.\cdots\). Notice the close correspondence between
the program and the formula: the functions \(f\) and \(g\) correspond to the predicates
\(F\) and \(G\), and function calls correspond to applications of the predicates;
an interested reader may wish to consult \cite{kobayashi19, kobayashi18} to learn
how program verification problems can be translated to \hfl{} formulas.
Our cyclic proof system for \hfl{} presented in this paper would, therefore, be useful for
semi-automated verification of higher-order programs.

A key issue in the design of our cyclic proof system is how to formalize
a decidable ``global trace condition,''  to guarantee the soundness of cyclic proofs
in the presence of higher-order predicates and alternating fixed-points.
Our global trace condition has been inspired by, and is actually very similar to that of Doumane~\cite{doumane}.
The decidability of our global trace condition is,
however, non-trivial, due to the presence of
higher-order predicates. {Inspired by the approach of Brotherston and Simpson~\cite{brotherston},
we reduce the global trace condition to the containment problem for B\"uchi automata.}

We also consider an infinitary version of our proof system, and prove that
the infinitary proof system is complete for sequents without higher-order variables.
The restriction to sequents without higher-order variables is sufficient for
the aforementioned application of our proof system to higher-order program verification.

The rest of this paper is structured as follows.
Section 2 reviews the syntax and semantics of \hfl{}.
Section 3 defines our cyclic proof system.
Sections 4 and 5 respectively prove the decidability of the global trace condition
and the soundness of our cyclic proof system.
Section 6 discusses the restricted form of completeness
of the infinitary variant of our proof system.
Section 7 discusses related work, and Section~8 concludes the paper.

\section{\hfl: Higher-Order Fixed-Point Arithmetic}
This section introduces the target logic \hfl, which is a higher-order logic with natural numbers and alternating fixed-points.
It has been introduced by Kobayashi et al.~\cite{kobayashi18,watanabe19} as an extension of higher-order modal {fixed-point} logic (HFL)~\cite{viswanathan04}.

\subsection{Syntax of \hfl} \label{sec: hfl}
\hfl is simply typed.
The syntax of \emph{types} is given as follows:
\begin{equation*}
  A ::= \mathbf{N} \mid T
  \qquad\qquad
  T ::= \mathbf{\Omega} \mid A \to T.
\end{equation*}
\( \mathbf{N} \) is the type of natural numbers, \( \mathbf{\Omega} \) is the type of propositions and \( A \to T \) is a function type.
Occurrences of \( \mathbf{N} \) are restricted to argument positions for a technical reason (see below).

Let $\mathcal{V}$ be a countably infinite set of \emph{variables}, ranged over by \( x, y, z, f, X, Y, Z, \dots \).
The syntax of \emph{terms} and \emph{formulas} is given by:
\begin{align*}
&\text{term} & s, t &\quad::=\quad x \mid \mathbf{Z} \mid \mathbf{S} t \\
&\text{formula} & \varphi, \psi &\quad::=\quad
s = t \mid \varphi \lor \psi \mid \varphi \land \psi %
\mid x \mid \lambda x^A. \varphi \mid \varphi \, \psi \mid \varphi \, t
\mid \mu x^T. \varphi \mid \nu x^T. \varphi.
\end{align*}
We shall often omit the type annotations.
The syntax of terms is standard: \( \mathbf{Z} \) represents zero and \( \mathbf{S} \) is the successor function.
A closed term must be of the form \( \mathbf{S}^n \mathbf{Z} \), which is often identified with the natural number \( n \).
Constructors of formulas are logical ones (\( s = t \), \( \varphi \lor \psi \) and \( \varphi \land \psi \)), those from the \( \lambda \)-calculus (variable \( x \), abstraction \( \lambda x. \varphi \), and application \( \varphi\,\psi \) and \( \varphi\,t \)), and fixed-point operators (least fixed-point \( \mu x. \varphi \) and greatest fixed-point \( \nu x. \varphi \)).
Some standard constructs such as summation \( t_1 + t_2 = s \), multiplication \( t_1 \times t_2 = s \), truth \( \top \) and quantifiers \( \forall \) and \( \exists \) are definable{; see examples later in this subsection}.
The set of \emph{free variables} is defined as usual; the binders are \( \lambda x \), \( \mu x \) and \( \nu x \).

A \emph{sequent} is a pair \( (\Gamma, \Delta) \) of finite sequences of formulas, written as \( \Gamma \vdash \Delta \).
{
For a finite sequence \(\Gamma \) of formulas,
\( \mathit{FV}(\Gamma) \) denotes
the union of the sets of free variables of formulas in \( \Gamma \).}
The free variables of a sequent \( \Gamma \vdash \Delta \) is \( \mathit{FV}(\Gamma) \cup \mathit{FV}(\Delta) \).

The type system is presented in \autoref{fig: typing-rules}.
Here \( \mathcal{H} \) is a \emph{type environment}, which is a map from a finite subset of variables to the set of types.
The typing rules should be easy to understand. %
The types of fixed-point formulas \( \mu x. \varphi \) and \( \nu y. \psi \), as well as the types of the variables \( x \) and \( y \),
are restricted to \(T\) (i.e., they cannot
 be \( \mathbf{N} \)).
A formula \( \varphi \) is \emph{well-typed} if \( \mathcal{H} \vdash \varphi : T \) for some \( \mathcal{H} \) and \( T \).
In the sequel, we shall consider only well-typed formulas.

For a sequent \( \Gamma \vdash \Delta \), we write \( \mathcal{H} \mid \Gamma \vdash \Delta \) if \( \mathcal{H} \vdash \varphi : \mathbf{\Omega} \) for every \( \varphi \) in \( \Gamma \) or \( \Delta \).
A sequent is \emph{well-typed} if \( \mathcal{H} \mid \Gamma \vdash \Delta \) for some \( \mathcal{H} \).

We give some examples of formulas and explain their intuitive meaning.
\begin{example}
  The truth \( \top \) and falsity \( \bot \) can be defined as fixed-points:
  \begin{equation*}
    \top ~:=~ \nu x^{\mathbf{\Omega}}. x
    \qquad\mbox{and}\qquad
    \bot ~:=~ \mu x^{\mathbf{\Omega}}. x.
  \end{equation*}
  The former means that \( \top \) is the greatest element in the set of truth values (i.e.~elements in the semantic domain of \( \mathbf{\Omega} \)) such that \( x = x \).
  Similarly \( \bot \) is the least truth value.
  The greatest \( \top_T \) and least \( \bot_T \) value of type \( T \) are defined in a similar way.
  \qed
\end{example}

\begin{example}
  The quantifiers over natural numbers are definable.
  Let \( \mathit{forall} \) be a predicate of type \( (\mathbf{N} \to \mathbf{\Omega}) \to \mathbf{N} \to \mathbf{\Omega} \) defined by
  \begin{equation*}
    \mathit{forall} \quad:=\quad \nu X. \lambda p^{\mathbf{N} \to \mathbf{\Omega}}. \lambda x^{\mathbf{N}}. p\,x \land X\,p\,(\mathbf{S}\,x).
  \end{equation*}
  Then \( \mathit{forall}\,\varphi\,n \) holds if and only if \( \varphi\,m \) holds for every \( m \ge n \).
  To see this, {notice} that \( (\mathit{forall}\,\varphi\,n) = (\varphi\,n) \land (\mathit{forall}\,\varphi\,(n+1)) \) since \( \mathit{forall} \) is a fixed-point.
  By iteratively applying this equation, we have
  \begin{equation*}
    (\mathit{forall}\,\varphi\,n) \quad=\quad
    (\varphi\,n)
    \land (\varphi\,(n+1))
    \land (\varphi\,(n+2))
    \land \cdots
  \end{equation*}
  and thus \( \mathit{forall}\,\varphi\,n \) if and only if \( \forall m \ge n. \varphi\,m \).\footnote{
    The reader may notice that this intuitive argument does not explain why we should use \( \nu \) instead of \( \mu \).
    Here we skip this subtle issue.
    An interested reader may compare the interpretations of \( \mathit{forall} \) and its \( \mu \)-variant, following the definition in the next subsection.
  }
  Then the universal quantifier over natural numbers is defined by
  \begin{equation*}
    \forall x. \varphi \quad:=\quad \mathit{forall}\,(\lambda x. \varphi)\,\mathbf{Z}.
  \end{equation*}
  The existential quantifier can be defined similarly.
  A direct definition is
  \begin{equation*}
    \exists x^\mathbf{N}. \varphi \quad:=\quad \Big(\mu Y^{\mathbf{N} \to \mathbf{\Omega}}. \lambda x. \varphi \lor Y \, (\mathbf{S}x) \Big) \, \mathbf{Z}.
  \end{equation*}
  The quantifiers over \( T \) are easier because of monotonicity (see next subsection); we define \( \forall x^T. \varphi := (\lambda x. \varphi)\,\bot_T \) and \( \exists x^T. \varphi := (\lambda x. \varphi)\,\top_T \).
  \qed
\end{example}

\begin{example}
  Let \( \mathit{sum} \) be the summation, i.e.~the predicate such that \( \mathit{sum}\,n\,m\,k \) holds if and only if \( n + m = k \).
  This predicate can be defined as follows:
  \begin{equation*}
    \mu \mathit{sum}.\,\lambda x^{\mathbf{N}}. \lambda y^{\mathbf{N}}. \lambda z^{\mathbf{N}}. (x = \mathbf{Z} \land y = z) \vee (\exists x'. \exists z'. x = \mathbf{S} x' \land \mathit{sum}\,x'\,y\,z' \land z = \mathbf{S} z').
  \end{equation*}
  This formula represents the standard inductive definition of the summation: \( \mathbf{Z} + y = y \) and \( x' + y = z' \Longrightarrow \mathbf{S} x' + y = \mathbf{S} z' \).
  One can define the multiplication in a similar way, representing the standard inductive definition of the multiplication using \( \mathit{sum} \).
  \qed
\end{example}

\begin{example}
  The inequality \( s < t \) on terms is also definable.
  The idea is to appeal {to} the following fact: if \( n < m \), then \( n + 1 = m \) or \( (n+1) < m \).
  This justifies
  \begin{equation*}
    (s < t) \quad:=\quad
    \big(\mu X. \lambda y^{\mathbf{N}}. (\mathbf{S} y = t) \lor X\,(\mathbf{S} y))\,s.
  \end{equation*}
  Here \( X \) is a variable of type \( \mathbf{N} \to \mathbf{\Omega} \) and \( X\,s' \) means \( s' < t \).
  Then the inequality \( s \neq t \) can be defined as \( (s < t) \lor (t < s) \).
  \qed
\end{example}

\begin{remark}
  \hfl does not have negation \( \neg \).
  The absence of negation plays an important role in the interpretation of fixed-point operators, as we shall see in the next subsection.
  For a formula \( \varphi \) with no free variables except for those of type \( \mathbf{N} \), the negation \( \neg \varphi \) can be obtained by replacing each logical connective with its De~Morgan dual,
  \begin{equation*}
    {\land} \leftrightsquigarrow {\lor},
    \qquad
    {\mu} \leftrightsquigarrow {\nu},
    \quad\mbox{and}\quad
    {=} \leftrightsquigarrow {\neq},
  \end{equation*}
  as discussed in~\cite{negation}.
  \qed
\end{remark}

\begin{figure}[t]
For terms:
\begin{center}  %
  \AxiomC{}
  \UnaryInfC{$\mathcal{H}, x:A \vdash x:A$}
  \DisplayProof
  \qquad
  \AxiomC{}
  \UnaryInfC{$\mathcal{H} \vdash \mathbf{Z}: \mathbf{N}$}
  \DisplayProof
  \qquad
  \AxiomC{$\mathcal{H} \vdash t: \mathbf{N}$}
  \UnaryInfC{$\mathcal{H} \vdash \mathbf{S}t: \mathbf{N}$}
  \DisplayProof
\end{center}

For formulas:
\begin{center}
  \AxiomC{$\mathcal{H} \vdash s: \mathbf{N}$}
  \AxiomC{$\mathcal{H} \vdash t: \mathbf{N}$}
  \BinaryInfC{$\mathcal{H} \vdash s=t: \mathbf{\Omega}$}
  \DisplayProof
\end{center}

\begin{center}
  \AxiomC{$\mathcal{H} \vdash \varphi: \mathbf{\Omega}$}
  \AxiomC{$\mathcal{H} \vdash \psi: \mathbf{\Omega}$}
  \BinaryInfC{$\mathcal{H} \vdash \varphi \lor \psi: \mathbf{\Omega}$}
  \DisplayProof
  \qquad
  \AxiomC{$\mathcal{H} \vdash \varphi: \mathbf{\Omega}$}
  \AxiomC{$\mathcal{H} \vdash \psi: \mathbf{\Omega}$}
  \BinaryInfC{$\mathcal{H} \vdash \varphi \land \psi: \mathbf{\Omega}$}
  \DisplayProof
\end{center}

\begin{center}
  \AxiomC{$\mathcal{H}, x: A \vdash \varphi: T $}
  \UnaryInfC{$\mathcal{H} \vdash \lambda x^A. \varphi: A \to T$}
  \DisplayProof
  \qquad
  \AxiomC{$\mathcal{H} \vdash \varphi: T' \to T$}
  \AxiomC{$\mathcal{H} \vdash \psi: T'$}
  \BinaryInfC{$\mathcal{H} \vdash \varphi \, \psi: T$}
  \DisplayProof
\end{center}

\begin{center}
  \AxiomC{$\mathcal{H} \vdash \varphi: \mathbf{N} \to T$}
  \AxiomC{$\mathcal{H} \vdash t: \mathbf{N}$}
  \BinaryInfC{$\mathcal{H} \vdash \varphi \, t: T$}
  \DisplayProof
\end{center}

\begin{center}
  \AxiomC{$\mathcal{H}, x:T \vdash \varphi:T$}
  \UnaryInfC{$\mathcal{H} \vdash \mu x^T. \varphi: T$}
  \DisplayProof
  \qquad
  \AxiomC{$\mathcal{H}, x:T \vdash \varphi:T$}
  \UnaryInfC{$\mathcal{H} \vdash \nu x^T. \varphi: T$}
  \DisplayProof
\end{center}
\caption{Typing Rules for \hfl.}
\label{fig: typing-rules}
\end{figure}

\subsection{Semantics of \hfl}
This subsection introduces the interpretations of types and formulas.

The interpretation of a type \( A \) is a poset \( \llbracket A \rrbracket = (\llbracket A \rrbracket, \le_A) \).
It is inductively defined by
\begin{align*}
  \llbracket \mathbf{N} \rrbracket &:= \mathbb{N}
  &
  x \le_{\mathbf{N}} y &:\Leftrightarrow x = y
  \\
  \llbracket \mathbf{\Omega} \rrbracket &:= \{\, \top, \bot \,\}
  &
  x \le_{\mathbf{\Omega}} y &:\Leftrightarrow x=\bot \mbox{ or } y=\top
  \\
  \llbracket A \to T \rrbracket &:= \{ f: \llbracket A \rrbracket \to \llbracket T \rrbracket \mid f \text{ is monotone} \}
  \quad
  &
  f \le_{A \to T} g &:\Leftrightarrow \forall x \in \llbracket A \rrbracket. f(x) \le_T g(x).
\end{align*}
Note that
\begin{itemize}
\item \( \llbracket A \to T \rrbracket \) is the space of \emph{monotone} functions, and
\item \( \llbracket T \rrbracket \) is a complete lattice for every \( T \).
\end{itemize}
On the contrary \( \llbracket A \rrbracket \) is not necessarily a complete lattice since \( \llbracket \mathbf{N} \rrbracket \) is not.

The interpretation \( \llbracket \mathcal{H} \rrbracket \) of a type environment \( \mathcal{H} \) is the set of mappings \( \rho \) such that \( \rho(x) \in \llbracket \mathcal{H}(x) \rrbracket \) for every \( x \) in the domain of \( \mathcal{H} \).
{The set \( \llbracket \mathcal{H} \rrbracket \) forms a poset with respect to}
the point-wise ordering.
An element of $\llbracket \mathcal{H} \rrbracket$ is called a \emph{valuation}.
We write \( \rho[x \mapsto v] \) for the mapping defined by \( \rho[x \mapsto v](x) = v \) and \( \rho[x \mapsto v](y) = \rho(y) \) for \( y \neq x \).

Assume \( \mathcal{H} \vdash \varphi : A \) (where \( \varphi \) is a formula or a term).
Its interpretation \( \llbracket \mathcal{H} \vdash \varphi : A \rrbracket \) is a monotone function \( \llbracket \mathcal{H} \rrbracket \longrightarrow \llbracket A \rrbracket \).
The definition is in \autoref{fig: definition-of-interpretation}.
Here \( \mathit{lfp}(f) \) and \( \mathit{gfp}(f) \) are the least and greatest fixed-points of the function \( f \).
The interpretations of the fixed-point operators are well-defined since every monotone function \( f : P \longrightarrow P \) on a complete lattice \( P \) has both the least and greatest fixed-points.
We shall write $\llbracket \mathcal{H} \vdash \varphi: A \rrbracket(\rho)$ simply $\llbracket \varphi \rrbracket_\rho$ when no confusion can arise.

\begin{figure}[t]
  \begin{align*}
    \llbracket \mathcal{H} \vdash x: A \rrbracket(\rho) &= \rho(x) \\
  \llbracket \mathcal{H} \vdash \mathbf{Z}: \mathbf{N} \rrbracket(\rho) &= 0 \\
  \llbracket \mathcal{H} \vdash \mathbf{S}t: \mathbf{N} \rrbracket(\rho) &= 1+\llbracket \mathcal{H} \vdash t: \mathbf{N} \rrbracket(\rho) \\
  \llbracket \mathcal{H} \vdash s = t: \mathbf{\Omega} \rrbracket(\rho) &= (\llbracket \mathcal{H} \vdash s: \mathbf{N} \rrbracket(\rho) = \llbracket \mathcal{H} \vdash t: \mathbf{N} \rrbracket (\rho)) \\
  \llbracket \mathcal{H} \vdash \varphi \lor \psi : \mathbf{\Omega}\rrbracket(\rho) &= (\llbracket \mathcal{H} \vdash \varphi: \mathbf{\Omega} \rrbracket(\rho) \lor \llbracket \mathcal{H} \vdash \psi: \mathbf{\Omega} \rrbracket(\rho)) \\
  \llbracket \mathcal{H} \vdash \varphi \land \psi: \mathbf{\Omega} \rrbracket(\rho) &= (\llbracket \mathcal{H} \vdash \varphi: \mathbf{\Omega} \rrbracket(\rho) \land \llbracket \mathcal{H} \vdash \psi: \mathbf{\Omega} \rrbracket(\rho)) \\
  \llbracket \mathcal{H} \vdash \lambda x^A. \varphi: A \to T \rrbracket(\rho) &= \lambda v \in \llbracket A \rrbracket. \llbracket \mathcal{H}, x:A \vdash \varphi: T \rrbracket(\rho[x \mapsto v]) \\
  \llbracket \mathcal{H} \vdash \varphi \, \psi: T \rrbracket(\rho) &= (\llbracket \mathcal{H} \vdash \varphi: A \to T \rrbracket(\rho)) \, (\llbracket \mathcal{H} \vdash \psi: A \rrbracket (\rho)) \\
  \llbracket \mathcal{H} \vdash \varphi \, t: T \rrbracket(\rho) &= (\llbracket \mathcal{H} \vdash \varphi: \mathbf{N} \to T \rrbracket(\rho)) \, (\llbracket \mathcal{H} \vdash t: \mathbf{N} \rrbracket (\rho)) \\
  \llbracket \mathcal{H} \vdash \mu x^T. \varphi: T \rrbracket(\rho) &= \text{lfp}(\llbracket \mathcal{H} \vdash \lambda x^T. \varphi: T \to T \rrbracket(\rho)) \\
  \llbracket \mathcal{H} \vdash \nu x^T. \varphi: T \rrbracket(\rho) &= \text{gfp}(\llbracket \mathcal{H} \vdash \lambda x^T. \varphi: T \to T \rrbracket(\rho))
  \end{align*}
  \caption{Interpretation of terms and formulas.}
  \label{fig: definition-of-interpretation}
\end{figure}

\begin{definition}
  Let $\mathcal{H} \mid \Gamma \vdash \Delta$ be a sequent and $\rho$ be a valuation in $\llbracket \mathcal{H} \rrbracket$.
  Then we write $\Gamma \models_\rho \Delta$
  if $\bigwedge_{\varphi \in \Gamma} \llbracket \varphi \rrbracket_\rho \leq \bigvee_{\psi \in \Delta} \llbracket \psi \rrbracket_\rho$,
  or equivalently, if \( \llbracket \varphi \rrbracket_\rho = \bot \) for some \( \varphi \in \Gamma \) or \( \llbracket \psi \rrbracket_{\rho} = \top \) for some \( \psi \in \Delta \).
  A sequent $\mathcal{H} \mid \Gamma \vdash \Delta$ is \emph{valid}
  if $\Gamma \models_\rho \Delta$ for all valuations $\rho$,
  and we denote it briefly by $\Gamma \models \Delta$.
\end{definition}

\section{A Cyclic Proof System for \hfl} \label{sec: cyclic}
In this section, we introduce a cyclic proof system for \hfl.
A cyclic proof is a proof diagram which can contain cycles and should satisfy a certain condition in order to ensure the soundness.
\autoref{subsec: def} describes derivations and the soundness condition to define cyclic proofs
and \autoref{subsec: rule} shows that the induction rule based on prefixed-points is admissible.
In \autoref{subsec: eg}, we show some examples of cyclic proofs.

\subsection{Definition of the Cyclic Proof System} \label{subsec: def}
\autoref{fig: ded-rules} shows our deduction rules, which are based on Gentzen's sequent calculus.
Here \( \varphi[\psi/x] \) {represents the} capture-avoiding substitution and $\Gamma[\varphi/x]$ represents the sequence of formulas which is the result of applying the
substitution $[\varphi/x]$ to all formulas in $\Gamma$.
The rules should be {easy to understand since most of the} rules are standard.  We explain uncommon rules.
The rule (Mono) is based on the fact
that each formula defines a monotone function, i.e.~\( \llbracket \psi \rrbracket \sqsubseteq \llbracket \chi \rrbracket \) implies \( \llbracket \varphi[\psi/x] \rrbracket \sqsubseteq \llbracket \varphi[\chi/x] \rrbracket \).  Since only a formula of type \( \mathbf{\Omega} \) can appear in a sequent, \( \llbracket \psi \rrbracket \sqsubseteq \llbracket \chi \rrbracket \) is expressed as \( \psi\,\vec{y} \vdash \chi\,\vec{y} \) for fresh \( \vec{y} \).
The rules (\( \lambda L \)) and (\( \lambda R \)) are justified by the fact that the \( \beta \)-equivalence \( (\lambda x. \varphi)\,\psi = \varphi[\psi/x] \) preserves semantics.
The rules ($\sigma L$) and ($\sigma R$) express the fact that \( \sigma x. \varphi \) ({where \( \sigma\)
is} \(\mu \) or \( \nu \)) is a fixed-point and thus \( \sigma x. \varphi = \varphi[\sigma x.\varphi/x] \).
The rule (Nat) says that {a variable of type \( \mathbf{N} \) indeed represents} a natural number, and (P1) and (P2) correspond to the axioms \( (\mathbf{Z} = \mathbf{S}x) \to \bot \) and \( (\mathbf{S}x = \mathbf{S}y) \to (x = y) \). %

\begin{figure}
\begin{itemize}
  \item Identity rules
\begin{center}
  \AxiomC{\mathstrut}
  \RightLabel{(Axiom)}
  \UnaryInfC{$\varphi \vdash \varphi$}
  \DisplayProof
  \qquad
  \AxiomC{$\Gamma \vdash \varphi, \Delta$}
  \AxiomC{$\Gamma, \varphi \vdash \Delta$}
  \RightLabel{(Cut)}
  \BinaryInfC{$\Gamma \vdash \Delta$}
  \DisplayProof
\end{center}

\item Structural rules
\begin{center}
  \AxiomC{$\Gamma \vdash \Delta$}
  \RightLabel{(Wk L)}
  \UnaryInfC{$\Gamma, \varphi \vdash \Delta$}
  \DisplayProof
  \qquad
  \AxiomC{$\Gamma \vdash \Delta$}
  \RightLabel{(Wk R)}
  \UnaryInfC{$\Gamma \vdash \varphi, \Delta$}
  \DisplayProof
\end{center}

\begin{center}
  \AxiomC{$\Gamma, \varphi, \varphi \vdash \Delta$}
  \RightLabel{(Ctr L)}
  \UnaryInfC{$\Gamma, \varphi \vdash \Delta$}
  \DisplayProof
  \qquad
  \AxiomC{$\Gamma \vdash \varphi, \varphi, \Delta$}
  \RightLabel{(Ctr R)}
  \UnaryInfC{$\Gamma \vdash \varphi, \Delta$}
  \DisplayProof
\end{center}

\begin{center}
  \AxiomC{$\Gamma, \psi, \varphi, \Gamma' \vdash \Delta$}
  \RightLabel{(Ex L)}
  \UnaryInfC{$\Gamma, \varphi, \psi, \Gamma' \vdash \Delta$}
  \DisplayProof
  \qquad
  \AxiomC{$\Gamma \vdash \Delta, \psi, \varphi, \Delta'$}
  \RightLabel{(Ex R)}
  \UnaryInfC{$\Gamma \vdash \Delta, \varphi, \psi, \Delta'$}
  \DisplayProof
\end{center}

\begin{center}
  \AxiomC{$\Gamma \vdash \Delta$}
  \RightLabel{(Subst)}
  \UnaryInfC{$\Gamma[\varphi/x] \vdash \Delta[\varphi/x]$}
  \DisplayProof
\end{center}

\begin{center}
  \AxiomC{$\Gamma, \psi \, \vec{y} \vdash \chi \, \vec{y}, \Delta$}
  \RightLabel{$\vec{y} \cap FV(\Gamma, \psi, \chi, \Delta) = \emptyset$, (Mono)}
  \UnaryInfC{$\Gamma, \varphi[\psi/x^T] \vdash \varphi[\chi/x^T], \Delta$}
  \DisplayProof
\end{center}

\item Logical rules ($\sigma = \mu, \nu$)

\begin{center}
  \AxiomC{$\Gamma[t/x, s/y] \vdash \Delta[t/x, s/y]$}
  \RightLabel{($= L$)}
  \UnaryInfC{$\Gamma[s/x, t/y], s = t \vdash \Delta[s/x, t/y]$}
  \DisplayProof
  \qquad
  \AxiomC{}
  \RightLabel{($= R$)}
  \UnaryInfC{$\Gamma \vdash t=t, \Delta$}
  \DisplayProof
\end{center}

\begin{center}
  \AxiomC{$\Gamma, \varphi \vdash \Delta$}
  \AxiomC{$\Gamma, \psi\vdash \Delta$}
  \RightLabel{($\lor L$)}
  \BinaryInfC{$\Gamma, \varphi \lor \psi\vdash \Delta$}
  \DisplayProof
  \qquad
  \AxiomC{$\Gamma \vdash \varphi, \psi, \Delta$}
  \RightLabel{($\lor R$)}
  \UnaryInfC{$\Gamma \vdash \varphi \lor \psi, \Delta$}
  \DisplayProof
\end{center}

\begin{center}
  \AxiomC{$\Gamma, \varphi, \psi\vdash \Delta$}
  \RightLabel{($\land L$)}
  \UnaryInfC{$\Gamma, \varphi \land \psi\vdash \Delta$}
  \DisplayProof
  \qquad
  \AxiomC{$\Gamma \vdash \varphi, \Delta$}
  \AxiomC{$\Gamma \vdash \psi, \Delta$}
  \RightLabel{($\land R$)}
  \BinaryInfC{$\Gamma \vdash \varphi \land \psi, \Delta$}
  \DisplayProof
\end{center}

\begin{center}
  \AxiomC{$\Gamma, \varphi[\psi/x] \, \vec{\psi} \vdash \Delta$}
  \RightLabel{($\lambda L$)}
  \UnaryInfC{$\Gamma, (\lambda x. \varphi) \, \psi \, \vec{\psi} \vdash \Delta$}
  \DisplayProof
  \qquad
  \AxiomC{$\Gamma \vdash \varphi[\psi/x] \, \vec{\psi}, \Delta$}
  \RightLabel{($\lambda R$)}
  \UnaryInfC{$\Gamma \vdash (\lambda x. \varphi) \, \psi \, \vec{\psi}, \Delta$}
  \DisplayProof
\end{center}

\begin{center}
  \AxiomC{$\Gamma, \varphi[\sigma x. \varphi/x] \, \vec{\psi} \vdash \Delta$}
  \RightLabel{($\sigma L$)}
  \UnaryInfC{$\Gamma, (\sigma x. \varphi) \, \vec{\psi} \vdash \Delta$}
  \DisplayProof
  \qquad
  \AxiomC{$\Gamma \vdash \varphi[\sigma x. \varphi/x] \, \vec{\psi}, \Delta$}
  \RightLabel{($\sigma R$)}
  \UnaryInfC{$\Gamma \vdash (\sigma x. \varphi) \, \vec{\psi}, \Delta$}
  \DisplayProof
\end{center}

\item Natural number rules

$N \equiv \mu X. \lambda x. (x = \mathbf{Z}) \lor (\exists x'. x = \mathbf{S}x' \land X \, x')$

\begin{center}
  \AxiomC{$\Gamma, N \, x^{\mathbf{N}} \vdash \Delta$}
  \RightLabel{(Nat)}
  \UnaryInfC{$\Gamma \vdash \Delta$}
  \DisplayProof
  \quad
  \AxiomC{\mathstrut}
  \RightLabel{(P1)}
  \UnaryInfC{$\mathbf{S}s = \mathbf{Z} \vdash$}
  \DisplayProof
  \quad
  \AxiomC{$\Gamma, s=t \vdash \Delta$}
  \RightLabel{(P2)}
  \UnaryInfC{$\Gamma, \mathbf{S}s = \mathbf{S}t \vdash \Delta$}
  \DisplayProof
\end{center}
\end{itemize}
\caption{Deduction Rules.}
\label{fig: ded-rules}
\end{figure}

Although {every leaf of an  ordinary proof tree is an} axiom,
this is not the case in cyclic proof systems.
A \emph{(finite) derivation tree} is a tree obtained by using the rules in \autoref{fig: ded-rules},
{whose} leaves are not necessarily axioms.
A leaf that is not an axiom is {called} \emph{open}.
  \begin{remark} \label{rem:mono}
  Though the rule (Mono) has only one premise in Figure~\ref{fig: ded-rules},
  from now on, we treat it as the following rule with multiple premises:
\begin{center}
  \AxiomC{$\{\Gamma, \psi \, \vec{y} \vdash \chi \, \vec{y}, \Delta\}_{1, \cdots, k}$}
  \RightLabel{$\vec{y} \cap FV(\Gamma, \psi, \chi, \Delta) = \emptyset$, (Mono)}
  \UnaryInfC{$\Gamma, \varphi[\psi/x^T] \vdash \varphi[\chi/x^T], \Delta$}
  \DisplayProof.
\end{center}
Here, $k$ is the number of occurrences of $x$ in $\varphi$,
and the rule above has a copoy of the premise $\Gamma, \psi \vec{y} \vdash \chi \vec{y}, \Delta$ for each occurence of $x$ in $\varphi$.
For example, the inference:
\begin{center}
  \AxiomC{$x \vdash x \lor y$}
  \RightLabel{(Mono)}
  \UnaryInfC{$x \lor \mu z. (x \lor z) \vdash (x \lor y) \lor \mu z. ((x \lor y) \lor z)$}
  \DisplayProof
\end{center}
should actually be interpreted as
\begin{center}
  \AxiomC{$x \vdash x \lor y$}
  \AxiomC{$x \vdash x \lor y$}
  \RightLabel{(Mono)}
  \BinaryInfC{$x \lor \mu z. (x \lor z) \vdash (x \lor y) \lor \mu z. ((x \lor y) \lor z)$}
  \DisplayProof.
\end{center}

This enables us to trace fixed-point operators precisely in a derivation tree. See Definition~\ref{def: trace-ho}.
\end{remark}

\begin{definition}[Pre-proof]
  \label{def: preproof}
  A \emph{pre-proof} consists of a finite derivation tree $\mathcal{D}$ and a function $\mathcal{R}$ that assigns to each open leaf $n$ in $\mathcal{D}$ a non-leaf node $\mathcal{R}(n)$ that has the same sequent as $ n $.
\end{definition}

A pre-proof induces the infinite derivation tree by iteratively replacing an open leaf \( n \) with \( \mathcal{R}(n) \).
This correspondence would be helpful to understand the definitions below.

A pre-proof is unsound {in general, i.e.,}
the root sequent of a pre-proof {may be invalid.} %
We introduce a {sanity} condition called the \emph{global trace condition},
and define cyclic proofs as pre-proofs that satisfy this condition.

Given a pre-proof \( (\mathcal{D}, \mathcal{R}) \), its \emph{path} is a (finite or infinite) sequence \( (n_i)_{i = 1, 2, \dots} \) of nodes of \( \mathcal{D} \) such that, for every \( i \),
\begin{itemize}
  \item if \( n_i \) is an open leaf, then \( n_{i+1} = \mathcal{R}(n_i) \), and
\item otherwise \( n_{i+1} \) is a premise of \( n_i \).
\end{itemize}

Given a path \( (n_i)_{i = 1, 2, \dots} \), {a \emph{pre-trace} in this path} is a sequence \( (\tau_i)_{i = 1, 2, \dots} \) of occurrences of formulas such that, for every \( i \), \( \tau_i \) is an occurrence of a formula
in the sequent \( n_i \) and \( \tau_{i+1} \) is a ``relevant occurrence'' of \( \tau_i \)
{in the sequent \(n_{i+1}\).}
{The latter condition means that $\tau_{i+1}$ originates from \(\tau_i\) in a bottom-up construction of the proof.}
For example, if \( n_i \) and \( n_{i+1} \) are respectively the conclusion and premise of (\(\wedge L\)),
  {i.e., if \(n_i\) is the sequent \(\Gamma, \varphi\land \psi \vdash \Delta\) and
\(n_{i+1}\) is \(\Gamma, \varphi, \psi \vdash \Delta\),
then (i) \(\varphi\) and \(\psi\) in \(n_{i+1}\) are relevant occurrences of \(\varphi\land \psi\)
  in \(n_i\), and (ii) each formula in \(\Gamma\) (resp. \(\Delta\)) of \(n_{i+1}\) is a relevant
  occurrence of the corresponding formula in \(\Gamma\) (resp. \(\Delta\)) of \(n_{i}\).}
  For another example, if \( n_i \) is the conclusion of (\(\vee L\)), \( n_{i+1} \) is the left premise and \( \tau_i \) is the occurrence of \( \varphi \lor \psi \), then \( \tau_{i+1} \) is the occurrence of \( \varphi \).
  The concrete definition, which we omit here, is lengthy but straightforward; a possible exception is (Mono), in which \( \psi\,\vec{y} \) and \( \chi\,\vec{y} \)
{  are defined as relevant occurrences of}
  \( \varphi[\psi/x] \) and \( \varphi[\chi/x] \)  respectively.
  {See Appendix~\ref{ap: rel-occurrence} for more detail.}
  A pre-trace is a \emph{trace} if, for infinitely many \( i \), \( \tau_i \) is the principal occurrence\footnote{An occurrence of a formula in the conclusion of a rule in \autoref{fig: ded-rules} is \emph{principal} if it
    {belongs to neither} \( \Gamma \) nor \( \Delta \).} of a logical rule.

  The global trace condition requires existence of a ``good'' trace for each infinite path.
  The appropriate notion of ``good'' traces depends on the logic.
  In \cite{brotherston}, a trace is ``good'' if it contains infinitely many principal occurrences of (\(\mu L\)) or (\(\nu R\)).
  In other words, a ``good'' trace contains infinitely many expansions of \( \mu \).
  This fairly simple condition comes from the restriction of usage of fixed-points: their logic does not allow alternation of fixed-points, e.g.~\( \mu P. \varphi \) is allowed only if the free predicate variables of \( \varphi \) are bound by \( \mu \).
  Allowing nested fixed-points makes the definition of ``good'' traces more complicated; the definition in \cite{doumane} refers to the most significant fixed-point operator among those that are expanded infinitely many times.
  The higher-order nature of \hfl{} requires us to more precisely track the usage of fixed-point operators.

  The following definition is inspired by the winning criterion of game semantics of \hfl{}~\cite{bruse14,kobayashi18,tsukada20}.
  {The idea is to track which occurences of fixed-point operators are unfolded infinitely in depth by annotating each occurrence with a sequence that grows with each unfolding.}
By abuse of notation, a path is written as a sequence \( (\Gamma_i \vdash \Delta_i)_{i = 1, 2, \dots} \) of sequents.
We often identify an occurrence of a formula with the formula.
For example, \( \tau_i \equiv \varphi \) means that \( \tau_i \) is an occurrence of \( \varphi \).

\begin{definition}[$\mu$-trace, $\nu$-trace] \label{def: trace-ho}
  Let $(\tau_i)_{i \geq 0}$ be a trace.
  We assign a sequence of natural numbers to every fixed-point operator in $\tau_i$ for all $i$ by the following algorithm.
  We use \( \sigma \) as a metavariable of fixed-point operators \( \{\, \mu,\nu \,\} \).
  We write \( \sigma_p \) for the fixed-point operator to which the sequence \( p \) is assigned.
  \begin{itemize}
    \item For every $\sigma$ in $\tau_0$, we assign $\epsilon$ to $\sigma$.
    \item
      If \( \tau_i \) is the principal occurrence of (\(\sigma L/R\)),
      \( \tau_i \) with annotation is \( \sigma_{p} x. \varphi \) (where fixed-point operators in \( \varphi \) are also annotated).
      Then \( \tau_{i+1} \) with annotation is \( \varphi[\sigma_{p.k} x.\varphi/x] \) where \( k \) is a natural number that has not been used in this annotation process.\footnote{{One can weaken the freshness requirement for \( k \): the minimal requirement is that the sequence \( p.k \) has not been used.}}
      \item
        If $\tau_i$ is the principal occurence of (Mono) and $\tau_{i+1}$ is in the $j$-th premises,
        $\tau_i$ with annotation is $\varphi[\psi_1/x_1, \cdots \psi_k/x_k]$ (see Remark~\ref{rem:mono}).
        Then $\tau_{i+1}$ with annotation is $\psi_j \vec{y}$.
    \item Otherwise, the sequence of a fixed-point operator in \( \tau_{i+1} \) comes from the corresponding operator in \( \tau_i \).
      For example, if \( \tau_i \) is the principal occurrence of (\(\lambda L/R\)) and \( \tau_i \) with annotation is \( (\lambda x. \varphi)\,\psi \), then \( \tau_{i+1} \) with annotation is \( \varphi[\psi/x] \).
  \end{itemize}
  Let $p[0:n]$ denote the sequence %
  consisting of the first $n$ elements %
  of $p$.
  We call $(\tau_i)_{i \geq 0}$ a \emph{$\mu$-trace (resp.~$\nu$-trace)}
  if there is an infinite sequence $p$ such that
  $p[0:n]$ is assigned to $\mu$ (resp.~$\nu$) in some $\tau_i$ for every natural number $n$.
\end{definition}
{In the definiton above, for each trace $(\tau_i)_{i \geq 0}$,
there is at most one infinite trace \(p\) that satisfies the condition above; hence,
no trace can be both a \(\mu\)-trace and a \(\nu\)-trace; see Lemma~\ref{lem: uniqueness}.}

\begin{example}\label{example:pre-proof}
  Let us consider the following pre-proof $(\mathcal{D}, \mathcal{R})$.
  \begin{center}
    \AxiomC{$(\star) \ \vdash (\nu f. \lambda g. g \, (f \, g)) \, (\mu x. \lambda a. a)$}
    \RightLabel{($\lambda R$)}
    \UnaryInfC{$\vdash (\lambda a. a) \, ((\nu f. \lambda g. g \, (f \, g)) \, (\mu x. \lambda a. a))$}
    \RightLabel{($\mu R$)}
    \UnaryInfC{$\vdash (\mu x. \lambda a. a) \, ((\nu f. \lambda g. g \, (f \, g)) \, (\mu x. \lambda a. a))$}
    \RightLabel{($\lambda R$)}
    \UnaryInfC{$\vdash (\lambda h.h\,((\nu f. \lambda g. g \, (f \, g))\,h))\,(\mu x. \lambda a. a)$}
    \RightLabel{($\nu R$)}
    \UnaryInfC{$(\star) \ \vdash (\nu f. \lambda g. g \, (f \, g)) \, (\mu x. \lambda a. a)$}
    \DisplayProof
  \end{center}
  In the diagram, the function \( \mathcal{R} \) is {indicated by the $\star$ marks:}
  the open leaf $(\star)$ is mapped to the other node marked \((\star)\).
  This pre-proof has a unique path, and the path has a unique trace $(\tau_i)_{i \geq 0}$ (since each sequent consists of a single formula).

  We assign {a sequence} of natural numbers to {each occurrence} of a fixed-point operator
  in the trace \( (\tau_i)_{i \geq 0} \).
    {Both} fixed-point operators in \( \tau_0 \) are annotated by the empty sequence:
    \begin{equation*}
      \tau_0
      \quad\equiv\quad
      (\nu_\epsilon f. \lambda g. g \, (f \, g)) \, (\mu_\epsilon x. \lambda a. a).
    \end{equation*}
    The first rule expands \( \nu \), and we annotate the recursive occurrence of this \( \nu \) with a fresh natural number, say \( 0 \):
    \begin{equation*}
      \tau_1
      \quad\equiv\quad
      (\lambda h.h\,((\nu_{0} f. \lambda g. g \, (f \, g))\,h))\,(\mu_{\epsilon} x. \lambda a. a).
    \end{equation*}
    The next rule is (\( \lambda R \)) and we just substitute the annotated formula \((\mu_{\epsilon} x. \lambda a. a)\) for \( h \):
    \begin{equation*}
      \tau_2
      \quad\equiv\quad
      (\mu_\epsilon x. \lambda a. a) \, ((\nu_{0} f. \lambda g. g \, (f \, g)) \, (\mu_\epsilon x. \lambda a. a)).
    \end{equation*}
    Then we expand \( \mu \) and annotate its recursive occurrences {(if there were any)} with \( 1 \);
    actually, since \( x \) does not appear in the body \( \lambda a.a \), the resulting formula does
    not have label \( 1 \).
    \begin{equation*}
      \tau_3
      \quad\equiv\quad
      (\lambda a. a) \, ((\nu_{0} f. \lambda g. g \, (f \, g)) \, (\mu_\epsilon x. \lambda a. a)).
    \end{equation*}
    Applying the \( \beta \)-reduction, we have
    \begin{equation*}
      \tau_4
      \quad\equiv\quad
      (\nu_{0} f. \lambda g. g \, (f \, g)) \, (\mu_\epsilon x. \lambda a. a).
    \end{equation*}
    The current node is the open leaf, and the next node is determined by \( \mathcal{R} \).
    The annotation is copied: \( \tau_5 \equiv \tau_4 \).
    The next rule is (\(\nu R\)) and we name the recursive occurrences \( 0.2 \), extending the annotation \( 0 \) by a fresh number \( 2 \):
    \begin{equation*}
      \tau_6
      \quad\equiv\quad
      (\lambda h.h\,((\nu_{0.2} f. \lambda g. g \, (f \, g))\,h))\,(\mu_{\epsilon} x. \lambda a. a).
    \end{equation*}
    By continuing this argument, we have
    \begin{equation*}
      \tau_{1 + 5k}
      \quad\equiv\quad
      (\lambda h.h\,((\nu_{p} f. \lambda g. g \, (f \, g))\,h))\,(\mu_{\epsilon} x. \lambda a. a)
    \end{equation*}
    where \( p = 0.2.4.\dots.(2k) \).
    Note that the annotation of \( \nu \) grows but that of \( \mu \) does not.
    Hence this trace is a \( \nu \)-trace but not a \( \mu \)-trace.
  \qed
\end{example}

A trace \( (\tau_i)_{i \geq 0} \) is a \emph{left trace} (resp.~\emph{right trace}) if \( \tau_0 \) occurs on the left (resp.~right) side of \( \vdash \).
Note that {every}  \( \tau_i \) occurs on the same side as \( \tau_0 \).
\begin{definition}[Cyclic proof]\label{def: cyclic-proof}
  A \emph{cyclic proof} %
  is a pre-proof that satisfies the \emph{global trace condition}: for every infinite path, %
  a tail of the path has a left $\mu$-trace or right $\nu$-trace.
\end{definition}

\begin{example}
  The pre-proof in Example~\ref{example:pre-proof} is a cyclic proof.
  \qed
\end{example}

\begin{remark}
  One may find our global trace condition (cf.~Definitions~\ref{def: trace-ho} and \ref{def: cyclic-proof}) complicated and wonder if it is possible to replace it with a simpler conidition such as the parity condition.
  A recent result~\cite[Theorem~25]{tsukada20} suggests a negative answer:
  It shows that the validity of \hfl formulas cannot be captured by parity games, but games with more complicated winning criteria.  Our global trace condition is inspired by the criteria.
  \qed
\end{remark}

\subsection{Some Admissible Rules} \label{subsec: rule}
{This subsection shows that some familiar rules for quantifiers and inductions are admissible in
our cyclic proof system.}

As we saw in \autoref{sec: hfl}, formulas with quantifiers can be expressed by fixed-points.
The proposition below enables us to use quantifier rules in our cyclic proof system.
\begin{proposition} \label{prop: quantifiers}
  If there is a cyclic proof %
  of $\Gamma \vdash \Delta$ derived by the rules in \autoref{fig: ded-rules} plus the following quantifier rules,
  then there exists a cyclic proof %
  of $\Gamma \vdash \Delta$ {without the quantifier rules.}

\begin{center}
  \AxiomC{$\Gamma, \varphi[\psi/x] \vdash \Delta$}
  \RightLabel{($\forall L$)}
  \UnaryInfC{$\Gamma, \forall x. \varphi \vdash \Delta$}
  \DisplayProof
  \qquad
  \AxiomC{$\Gamma \vdash \varphi, \Delta$}
  \RightLabel{$x \not \in FV(\Gamma, \Delta)$ ($\forall R$)}
  \UnaryInfC{$\Gamma \vdash \forall x. \varphi, \Delta$}
  \DisplayProof
\end{center}

\begin{center}
  \AxiomC{$\Gamma, \varphi \vdash \Delta$}
  \RightLabel{$x \not \in FV(\Gamma, \Delta)$ ($\exists L$)}
  \UnaryInfC{$\Gamma, \exists x. \varphi \vdash \Delta$}
  \DisplayProof
  \qquad
  \AxiomC{$\Gamma \vdash \varphi[\psi/x], \Delta$}
  \RightLabel{($\exists R$)}
  \UnaryInfC{$\Gamma \vdash \exists x. \varphi, \Delta$}
  \DisplayProof
\end{center}
\end{proposition}
\begin{proof}
  \iffull
  See \autoref{ap: quantifiers}.
  \else
  See \cite{KoriCSLfull}.
  \fi
\end{proof}

We can also embed explicit induction rules, so-called Park's fixed-point rules:
\begin{center}
  \AxiomC{$\Gamma, \varphi[\chi/x] \, \vec{y} \vdash \chi \, \vec{y}, \Delta$}
  \AxiomC{$\Gamma, \chi \, \vec{\psi} \vdash \Delta$}
  \RightLabel{$\vec{y} \cap FV(\Gamma,  \varphi[\chi/x], \Delta) = \emptyset$ (Pre)}
  \BinaryInfC{$\Gamma, (\mu x. \varphi) \, \vec{\psi} \vdash \Delta$}
  \DisplayProof
\end{center}
\begin{center}
  \AxiomC{$\Gamma, \chi \, \vec{y} \vdash \varphi[\chi/x] \, \vec{y}, \Delta$}
  \AxiomC{$\Gamma \vdash \chi \, \vec{\psi}, \Delta$}
  \RightLabel{$\vec{y} \cap FV(\Gamma,  \varphi[\chi/x], \Delta) = \emptyset$ (Post)}
  \BinaryInfC{$\Gamma \vdash (\nu x. \varphi) \, \vec{\psi}, \Delta$}
  \DisplayProof
\end{center}
They are inspired by Knaster-Tarski's fixed-point theorem: these rules replace pre/postfixed-points
{with} least/greatest fixed-points.
The rule (Pre) is sound because $\chi$ is a prefixed-point by the left {premise}
and $\mu x. \varphi$ is the least one.
The same argument holds for (Post).
\begin{proposition} \label{prop: explicit}
  If there is a cyclic proof %
  of $\Gamma \vdash \Delta$ derived by the rules in \autoref{fig: ded-rules} $+$ (Pre) $+$ (Post),
  then there exists a cyclic proof %
  of $\Gamma \vdash \Delta$ {without (Pre) and (Post).}
\end{proposition}
\begin{proof}
 { Given a cyclic proof \(\Pi\) that may use (Pre) and (Post),
  we first construct a pre-proof $\Pi'$ by removing (Pre) and (Post).
  For every instance of (Pre) of the form:}
  \begin{center}
    \AxiomC{$\Gamma, \varphi[\chi/x] \, \vec{y} \vdash \chi \, \vec{y}, \Delta$}
    \AxiomC{$\Gamma, \chi \, \vec{\psi} \vdash \Delta$}
    \RightLabel{(Pre)}
    \BinaryInfC{$\Gamma, (\mu x. \varphi) \, \vec{\psi} \vdash \Delta$}
    \DisplayProof
  \end{center}
  we will replace it with the following diagram.
  {\footnotesize
  \begin{center}
    \AxiomC{$\Pi$}
      \AxiomC{$\Gamma, \chi \, \vec{\psi} \vdash \Delta$}
      \RightLabel{(Wk L)}
      \UnaryInfC{$\Gamma, (\mu x. \varphi) \, \vec{\psi}, \chi \, \vec{\psi} \vdash \Delta$}
      \RightLabel{(Cut)}
    \BinaryInfC{$\Gamma, (\mu x. \varphi) \, \vec{\psi} \vdash \Delta$}
    \DisplayProof
  \end{center}

  \begin{center}
  $\Pi := $
    \AxiomC{$\Gamma, \varphi[\chi/x] \, \vec{y} \vdash \chi \, \vec{y}, \Delta$}
    \RightLabel{(Wk L)}
    \UnaryInfC{$\Gamma, (\mu x. \varphi) \, \vec{y}, \varphi[\chi/x] \, \vec{y} \vdash \chi  \, \vec{y}, \Delta$}
      \AxiomC{$(\star) \ \Gamma, (\mu x. \varphi) \, \vec{y} \vdash \chi  \, \vec{y}, \Delta$}
      \RightLabel{(Mono)}
      \UnaryInfC{$\Gamma, \varphi[\mu x. \varphi/x] \, \vec{y} \vdash \varphi[\chi/x] \, \vec{y}, \Delta$}
      \RightLabel{(Wk R)}
      \UnaryInfC{$\Gamma, \varphi[\mu x. \varphi/x] \, \vec{y} \vdash \chi  \, \vec{y}, \varphi[\chi/x] \, \vec{y}, \Delta$}
      \RightLabel{($\mu$L)}
      \UnaryInfC{$\Gamma, (\mu x. \varphi) \, \vec{y} \vdash \chi  \, \vec{y}, \varphi[\chi/x] \, \vec{y}, \Delta$}
    \RightLabel{(Cut)}
    \BinaryInfC{$(\star) \ \Gamma, (\mu x. \varphi) \, \vec{y} \vdash \chi  \, \vec{y}, \Delta$}
    \RightLabel{(Subst)}
    \UnaryInfC{$\Gamma, (\mu x. \varphi) \, \vec{\psi} \vdash \chi  \, \vec{\psi}, \Delta$}
    \DisplayProof
  \end{center}
  }
  (Post) can also be removed in the same manner.

  We show that the resulting pre-proof $\Pi'$ is a cyclic proof.
  For each infinite path $\pi$ in $\Pi'$,
  if some tail of $\pi$ goes through only one cycle as the above one from $(\star)$ to $(\star)$
  then we can trace the left $\mu$ in $\mu x. \varphi$ or the right $\nu$ in $\nu x. \varphi$.
  Otherwise, there exists a corresponding infinite path in $\Pi$, which satisfies the global trace condition.
  Therefore $\Pi'$ is a cyclic proof of $\Gamma \vdash \Delta$.
\end{proof}

\subsection{Examples} \label{subsec: eg}
This subsection presents two examples of cyclic proofs.

\begin{example}[Well-foundedness of a tree]
  In this example, we write %
  $\mathbf{N}^*$ for the type of finite sequences of natural numbers.
  We also use the concatenation operation \( ({-}) \cdot ({-}) \).
  This $\mathbf{N}^*$ can be expressed by $\mathbf{N}$ and thus this additional type does not increase the expressivity.

  {A \emph{tree} is a subset of finite sequences of natural numbers that represents the complement of the tree;
  the idea is to regard \( \epsilon \) as the root and \( p \) as the parent of \( p \cdot i \).}
  Let $f^{\mathbf{N}^* \to \mathbf{\Omega}}$ be a term representing a tree.
  We define $\Phi$ and $\Psi$ as follows:
  \begin{align*}
    \Phi &:= \mu w^{(\mathbf{N^*} \to \mathbf{\Omega}) \to \mathbf{\Omega}}. \lambda k^{\mathbf{N^*} \to \mathbf{\Omega}}. k \, \epsilon \lor \forall i^\mathbf{N}. w \, (\lambda z^{\mathbf{N}^*}. k \, (i \cdot z)) \\
    \Psi &:= \mu v^{\mathbf{N}^* \to \mathbf{\Omega}}. \lambda z^{\mathbf{N}^*}. f \, z \lor \forall i^{\mathbf{N}}. v \, (z \cdot i)
  \end{align*}
  Then both $\Phi \, f$ and $\Psi \, \epsilon$ represent well-foundedness of $f$, i.e.~whether there is no infinite path in $f$.
  $\Phi$ checks whether the tree $k$ satisfies well-foundedness.
  This returns true if $k$ has only one node or
all the immediate children of the root satisfy well-foundedness.
  $\Psi$ checks whether the subtree of $f$ whose root node is $z$ satisfies well-foundedness.
  This returns true if $z$ is a leaf or all subtrees under $z$ satisfy well-foundedness.

  A cyclic proof of $\Phi \, f \vdash \Psi \, \epsilon$ is given as follows:

  \begin{center}
  {\footnotesize
    \AxiomC{}
    \RightLabel{(Axiom)}
    \UnaryInfC{$f \, l \vdash f \, l$}
      \AxiomC{$(\dagger) \quad \Phi \, (\lambda z. f \, (l \cdot n \cdot z)) \vdash \Psi \, (l \cdot n)$}
      \AxiomC{$(\star) \quad Y_\Phi \, (\mathbf{S}n) \vdash Y_\Psi \, (\mathbf{S}n)$}
      \RightLabel{($\land L, R$)}
      \BinaryInfC{$\Phi \, (\lambda z. f \, (l \cdot n \cdot z)) \land Y_\Phi \, (\mathbf{S}n) \vdash \Psi \, (l \cdot n) \land Y_\Psi \, (\mathbf{S}n)$}
      \RightLabel{($\nu L, R$)}
      \UnaryInfC{$(\star) \quad Y_\Phi \, n \vdash Y_\Psi \, n$}
      \RightLabel{(Subst)}
      \UnaryInfC{$Y_\Phi \, \mathbf{Z} \vdash Y_\Psi \, \mathbf{Z}$}
      \RightLabel{($\lor L, \lor R, Wk$)}
    \BinaryInfC{$f \, l \lor \forall i. \Phi \, (\lambda z. f \, (l \cdot i \cdot z)) \vdash f \, l \lor \forall i. \Psi \, (l \cdot i)$}
    \RightLabel{($\mu L, R$)}
    \UnaryInfC{$(\dagger) \quad \Phi \, (\lambda z. f \, (l \cdot z)) \vdash \Psi \, l$}
    \RightLabel{(Subst) and $f \equiv \lambda z. f \, (\epsilon \cdot z)$}
    \UnaryInfC{$\Phi \, f \vdash \Psi \, \epsilon$}
    \DisplayProof}
  \end{center}
  where $Y_\Phi \equiv (\nu Y. \lambda i. \Phi \, (\lambda z. f \, (l \cdot i \cdot z)) \land Y \, (\mathbf{S}i))$
  and $Y_\Psi \equiv (\nu Y. \lambda i. \Psi \, (l \cdot i) \land Y \, (\mathbf{S}i))$.
  Here we omit (Subst) for open leaves;
  hence cycles of $(\star)$ and $(\dagger)$ are valid, although two nodes for each label have different sequents.
  Note that $Y_\Phi \, \mathbf{Z} \equiv \forall i. \Phi \, (\lambda z. f \, (l \cdot i \cdot z))$
  and $Y_\Psi \equiv \forall i. \Psi \, (l \cdot i)$.

  For all infinite paths,
  if the path includes $(\dagger) \to (\dagger)$ infinitely
  then we can trace the left $\mu$ in $\Phi$ and otherwise we can trace the right $\nu$ in $\forall i. \Psi \, (l \cdot i)$.
  \qed
\end{example}

The example below demonstrates an application of our cyclic proof system
to program verification, based on the reduction of
Kobayashi et al.~\cite{kobayashi18, watanabe19} from program verification to
\hfl{} validity checking.
\begin{example}[Example 2.4 and 3.3 in~\cite{watanabe19}] \leavevmode
  Consider the following OCaml-like program.
  \begin{verbatim}
    let rec repeat f x =
      if x = 0 then ()
      else if * then repeat f (f x) else repeat f (x-1)
    in let y = input() in
      repeat (fun x -> x-y) n
  \end{verbatim}
  In this program, \verb_*_ represents a non-deterministic Boolean value and
  \verb_input()_ means   %
  a user input.
  We aim to verify that an %
  appropriate input \verb|y| makes this program eventually terminate.

  To verify this, we have to check $\vdash \text{input} \, 0 \, (g \, n)$ where
  \begin{align*}
    \text{repeat} &:= \mu R. \lambda f. \lambda x. (x = 0) \lor (\exists x'. x = x'+1 \land f \, x \, (R \, f) \land R \, f \, x') \\
    \text{sub} &:= \lambda y. \lambda x. \lambda k. k \, (x-y) \\
    g &:= \lambda z. \lambda y. \text{repeat} \, (\text{sub} \, y ) \, z \\
    \text{input} &:= \mu I. \lambda x. \lambda k. (k \, x) \lor (I \, (x+1) \, k)
  \end{align*}
  where $(-)$ is defined naturally by using $\mu$.
{Here, functions on integers of type $\mathbf{N} \to \mathbf{N}$ in the program have been turned into
  predicates of type $\mathbf{N} \to (\mathbf{N} \to \mathbf{\Omega}) \to \mathbf{\Omega}$,
  which are obtained by CPS translation; for example, the function
  \texttt{fun x->x-y} has been turned into \(\text{sub}\,y\ (\equiv \lambda x.\lambda k.k(x-y))\).
  Note that \(\text{input} \, 0 \, (g \, n)\) is equivalent to
  \(\exists y.g\,n\,y\), which models an angelic non-determinsm of \texttt{input()} in the program.
}

  The goal sequent can be proved by the following cyclic proof:
  \begin{center}
    \AxiomC{}
    \RightLabel{($=$R)}
    \UnaryInfC{$\vdash 0=0$}
    \RightLabel{($\mu$R, $\lor$R)}
    \UnaryInfC{$\vdash \text{repeat} \, (\text{sub} \, 1) \, 0$}
    \RightLabel{($=$L)}
    \UnaryInfC{$n=0 \vdash \text{repeat} \, (\text{sub} \, 1) \, n$}
      \AxiomC{$\Pi$}
      \RightLabel{($\mu$L)}
      \BinaryInfC{$(\star, \dagger) \ \mathbf{N} \, n \vdash \text{repeat} \, (\text{sub} \, 1) \, n$}
    \RightLabel{($\lambda$R, Nat)}
    \UnaryInfC{$\vdash g \, n \, 1$}
    \RightLabel{(Wk R)}
    \UnaryInfC{$\vdash g \, n \, 1, \text{input} \, 2 \, (g \, n)$}
    \RightLabel{($\mu$R, $\lor$R)}
    \UnaryInfC{$\vdash \text{input} \, 1 \, (g \, n)$}
    \RightLabel{(Wk R)}
    \UnaryInfC{$\vdash g \, n \, 0, \text{input} \, 1 \, (g \, n)$}
    \RightLabel{($\mu$R, $\lor$R)}
    \UnaryInfC{$\vdash \text{input} \, 0 \, (g \, n)$}
    \DisplayProof
  \end{center}
  \begin{center}
    $\Pi := $
    \AxiomC{$(\star) \ \mathbf{N} \, n' \vdash \text{repeat} \, (\text{sub} \, 1) \, n'$}
    \RightLabel{($\lambda$R)}
    \UnaryInfC{$\mathbf{N} \, n' \vdash (\text{sub} \, 1) \, (n'+1) \, (\text{repeat} \, (\text{sub} \, 1))$}
    \AxiomC{$(\dagger) \ \mathbf{N} \, n' \vdash \text{repeat} \, (\text{sub} \, 1) \, n'$}
    \RightLabel{($\land$R)}
    \BinaryInfC{$\mathbf{N} \, n' \vdash (\text{sub} \, 1) \, (n'+1) \, (\text{repeat} \, (\text{sub} \, 1)) \land \text{repeat} \, (\text{sub} \, 1) \, n'$}
    \RightLabel{($\mu$R)}
    \UnaryInfC{$\mathbf{N} \, n' \vdash \text{repeat} \, (\text{sub} \, 1) \, (n'+1)$}
    \RightLabel{($=$L)}
    \UnaryInfC{$\mathbf{N} \, n', n=n'+1 \vdash \text{repeat} \, (\text{sub} \, 1) \, n$}
    \RightLabel{($\exists$L, $\land$L)}
    \UnaryInfC{$\exists n'. \mathbf{N} \, n' \land n=n'+1 \vdash \text{repeat} \, (\text{sub} \, 1) \, n$}
    \DisplayProof
  \end{center}

  This satisfies the global trace condition because we can trace the left $\mu$ in $\mathbf{N}$ for every infinite path.
  \qed
\end{example}

\section{Decidability of the Global Trace Condition}
{For our cyclic proof system to be useful,
there should}
exist an algorithm to check whether a given proof candidate is indeed a cyclic proof.
{It is easy to check
whether a given candidate is a pre-proof,
but it is non-trivial to check whether the pre-proof also satisfies
the global trace condition.}
In this section %
we prove %
that the global condition is indeed decidable.
We follow the approach in \cite{brotherston}, which reduces the global trace condition of a pre-proof to \buchi automata containment.
One automaton accepts all infinite paths of the pre-proof and the other accepts
those that satisfy %
the global trace condition.
Then whether the pre-proof is a cyclic proof corresponds to the %
inclusion between these automata.

Let us briefly recall the definition of \buchi automata.
A \emph{(nondeterministic) B\"uchi automaton} is a tuple $(Q, \Sigma, \delta, Q_0, F)$ where
$Q$ is a finite set of elements called \emph{states},
$\Sigma$ is a finite set of symbols,
$\delta: Q \times \Sigma \times Q$ is a \emph{transition relation},
$Q_0 \subseteq Q$ is a set of \emph{initial states}, and
$F \subseteq \delta$ is a set of accepting transition rules.\footnote{This differs from the standard definition, in which the acceptance condition is specified by the set of accepting \emph{states}.  It is not difficult to see that this change does not affect the expressive power.}
Given an infinite word $w = (a_i)_{i \geq 0} \in \Sigma^\omega$,
a \emph{run} over $w$ is a sequence $(q_i)_{i \geq 0}$ of states such that $q_0 \in Q_0$ and $(q_i, a_i, q_{i+1}) \in \delta$ for all $i \geq 0$.
This run is \emph{accepting} if \( (q_i, a_i, q_{i+1}) \in F \) for infinitely many \( i \).
  The automaton \emph{accepts} an infinite word if there is an accepting run over the word.

Let \( (\mathcal{D}, \mathcal{R}) \) be a given pre-proof.  The alphabet \( \Sigma \) is the set of nodes of \( \mathcal{D} \).
Then an infinite path is represented as an infinite word, and it is easy to construct an automaton \( \mathcal{A}_{\textit{path}} \) that accepts all the paths of \( \mathcal{D} \).
We define an automaton \( \mathcal{A}_{\textit{gtc}} \) that checks if a given path satisfies the global trace condition.

The idea of the automaton is as follows.
Recall Definition~\ref{def: trace-ho}, in which we assign a sequence of natural numbers to each occurrence of a fixed-point operator.
One cannot directly simulate the annotation process by an automaton since the set of sequences of natural numbers is infinite.
However, at the end of Definition~\ref{def: trace-ho}, we focus on an infinite sequence \( p \) of natural numbers, and sequences that are not prefixes of \( p \) can be safely ignored.
Furthermore, it suffices to remember exactly one finite sequence by the following argument.
  Consider the case that \( \tau_i \equiv \sigma_{q} x. \varphi \) where \( q \) is a prefix of \( p \), and \( \tau_{i+1} \equiv \varphi[\sigma_{q.k} x. \varphi/x] \).
  \begin{itemize}
  \item
    If \( q.k \) is not a prefix of \( p \), then we can safely forget the annotation \( q.k \).
  \item
    If \( q.k \) is a prefix of \( p \), then we can safely forget the annotation \( q \).
  \end{itemize}
  To see the latter, observe that other expansions of \( \sigma_{q} \) generate annotations \( q.k' \) with \( k' \neq k \) by freshness of \( k' \).
  Hence \( q.k' \) is not a prefix of \( p \) and thus we can forget it.
  So any extension of \( q \) that will be generated afterward can be safely ignored, and thus we can forget \( q \) itself.

The above argument motivates the following definition.
\begin{definition}
  A \emph{marked formula} \( \check{\varphi} \) is a formula in which some occurrences of fixed-point operators \( \sigma \) are \emph{marked}; a marked fixed-point operator is written as \( \sigma_{\bullet} \).
  We write \( |\check{\varphi}| \) for the (standard) formula obtained by removing marks.
  An \emph{occurrence-with-marks} \( \check{\tau} \) is a pair \( (\tau, \check{\varphi}) \) of an occurrence \( \tau \) and a marked formula \( \check{\varphi} \) such that \( \tau \equiv |\check{\varphi}| \).
\end{definition}

  We define \( \mathcal{A}_{\mathit{gtc}} \).
The set of states of \( \mathcal{A}_{\mathit{gtc}} \) consists of occurrences-with-marks of \( \mathcal{D} \) and a distinguished (initial) state \( \ast \).
Most rules just simulate Definition~\ref{def: trace-ho}.
For example, if \( \tau \) is the principal occurrence in the conclusion of (\( \lambda L \)), \( (\lambda x. \check{\varphi})\,\check{\psi} \) is a marked formula such that \( \tau \equiv (\lambda x. |\check{\varphi}|)\,|\check{\psi}| \) and \( n \) is the premise, then \( ((\tau, (\lambda x. \check{\varphi})\,\check{\psi}), n, (\tau', \check{\varphi}[\check{\psi}/x])) \) where \( \tau' \) is the unique occurrence in \( v \) that is relevant to \( \tau \).
Important transition rules are those dealing with (\( \sigma L/R \)).
Consider the state \( \check{\tau} = (\tau, \check{\varphi}) \) where \( \tau \) is the principal occurrence of (\( \sigma L/R \)).
Let \( n \) be the premise and \( \tau' \) be the unique relevant occurrence in \( n \).
\begin{itemize}
\item Case \( \check{\varphi} \equiv (\sigma x. \check{\psi})\,\check{\vec{\chi}} \):
  The automaton just unfolds the fixed-point operator, i.e.,
  \begin{equation*}
    ((\tau, (\sigma x. \check{\psi})\,\check{\vec{\chi}}),~ n,~ (\tau', \check{\psi}[(\sigma x. \check{\psi})/x]\,\check{\vec{\chi}})) \in \delta.
  \end{equation*}
\item Case \( \check{\varphi} \equiv (\sigma_{\bullet} x. \check{\psi})\,\check{\vec{\chi}} \):
  {Note that there may be other copies of \( \sigma_{\bullet} x. \check{\psi} \) in \( \check{\vec{\chi}} \) or \( \check{\psi} \).
  So the automaton has to choose which copy should be tracked, and the transition is nondeterministic.
  If the automaton decides to track the occurrence of \( \sigma_{\bullet} x. \check{\psi} \) being unfolded, it removes all marks in \( \check{\vec{\chi}} \) and tracks the recursive calls of the head occurrence by marking them:
  \begin{equation*}
    ((\tau, (\sigma_{\bullet} x. \check{\psi})\,\check{\vec{\chi}}),~ n,~ (\tau', |\check{\psi}|[(\sigma_{\bullet} x. |\check{\psi}|)/x]\,|\check{\vec{\chi}}|))
    \in \delta.
  \end{equation*}
  Otherwise it removes the mark of the fixed-point operator being unfolded and unfolds it:
  \begin{equation*}
    ((\tau, (\sigma_{\bullet} x. \check{\psi})\,\check{\vec{\chi}}),~ n,~ (\tau', \check{\psi}[(\sigma x. \check{\psi})/x]\,\check{\vec{\chi}})),
    \in \delta.
  \end{equation*}
  If the rule is (\( \mu L \)) or (\( \nu R \)), then the former transition is an accepting transition.
  All other transitions for (\( \sigma L/R \)) and other rules are not accepting.
  }
\end{itemize}
The initial state either ignores the input (\( (\ast, n, \ast) \in \delta \)) or nondeterministically chooses an occurrence-with-marks (\( (\ast, n, \check{\tau}) \in \delta \) if \( \check{\tau} \) is an occurrence in \( n \) and has exactly one marked fixed-point operator).

\begin{lemma}
  Let $(\mathcal{D}, \mathcal{R})$ be a pre-proof.
  For every infinite path $\pi$,
  $\pi \in \mathcal{L}(\mathcal{A}_{\mathit{gtc}})$ if and only if $\pi$ satisfies the global trace condition.
\end{lemma}
\begin{proof}
  Assume that \( \pi \) satisfies the global trace condition.
  Let \( ((\tau_i)_{i \geq 0}, p) \) be the pair of a trace (with annotations) and an infinite sequence of natural numbers that witnesses the global trace condition.
  For each \( i \), let \( q_i \) be the longest prefix of \( p \) among the annotations in \( \tau_i \).
  Let us mark \( \sigma_{q_i} \) in \( \tau_i \) and write \( \check{\varphi}_i \) for the resulting marked formula.
  Then \( (\tau_i, \check{\varphi}_i)_{i \geq 0} \) is an accepting run.

  We prove the converse.
  Assume a (possibly non-accepting) run, which determines a trace \( (\tau_i)_{i \geq 0} \).
  We annotate \( (\tau_i)_{i \geq 0} \) following Definition~\ref{def: trace-ho}.
  Let \( q_i \) be the sequence assigned to the marked operator in \( \tau_i \) and \( p \) be the limit of \( (q_i)_{i \geq 0} \).
  The transition rules ensure the well-definedness of \( q_i \) and \( p \).
  If the run is accepting, \( p \) is infinite since \( (q_i)_{i \geq 0} \) must grow infinitely many times.
  Furthermore it is a left \( \mu \)-trace or right \( \nu \)-trace as all accepting transitions are for (\( \mu L \)) or (\( \nu R \)).
\end{proof}

We have the following theorem as a corollary.
\begin{theorem}
  The validity checking of a pre-proof $(\mathcal{D}, \mathcal{R})$ is decidable.
\end{theorem}

\section{Soundness}
The soundness proof follows the proof strategy of \cite{brotherston}.
  We prove the claim by contraposition.
  Assume that the conclusion \( \Gamma \vdash \Delta \) of a cyclic proof is invalid.
  Then there exists a valuation \( \rho \) such that \( \Gamma \not\models_{\rho} \Delta \).
  Since all rules are \emph{locally sound} (i.e.~if the premises are valid under a valuation, the conclusion is also valid under the valuation),
  there exists an infinite path whose sequents are invalid under \( \rho \).
  Using the global trace condition, we construct an infinite decreasing chain of ordinals, a contradiction.

To be precise, we impose on the infinite path a condition slightly stronger than the invalidity under \( \rho \).
  To describe the condition, we need fixed-point operators \( \mu^{\alpha} x. \varphi \) and \( \nu^{\alpha} x. \varphi \) annotated by ordinals.
  The semantics of \( \mu^{\alpha} x. \varphi \) is given by
  \begin{equation*}
    \llbracket{\mu^{0} x^T. \varphi}\rrbracket(\rho) := \bot_T
    \quad\mbox{and}\quad
    \llbracket{\mu^{\alpha} x^T. \varphi}\rrbracket(\rho) := \bigsqcup_{\beta < \alpha} \llbracket{(\lambda x. \varphi)\,(\mu^{\beta} x^T. \varphi)}\rrbracket(\rho).
  \end{equation*}
  The definition of \( \nu^{\alpha} x.\varphi \) is similar (but uses the dual operations).

  During the construction of the path, we give ordinal annotations to occurrences \( \mu \) on the left side of \( \vdash \) and to occurrences of \( \nu \) on the right side.
  The annotation processes are essentially the same as that in Definition~\ref{def: trace-ho}.
  Again, the most important case is that \( \tau_i \) is the principal occurrence of (\( \mu L \)) or (\( \nu R \)).
  Consider the (\( \mu L \)) case.
  Then \( \tau_i \equiv (\mu^{\alpha} x. \varphi)\,\vec{\psi} \) (where \( \mu \) in \( \varphi \) and \( \vec{\psi} \) are also annotated).
  By the assumption of the invalidity of the sequent under \( \rho \), we have \( \llbracket(\mu^{\alpha} x. \varphi)\,\vec{\psi}\rrbracket(\rho) = \top \).
  The annotation to \( \tau_{i+1} \) is \( \varphi[\mu^{\beta} x. \varphi/x]\,\vec{\psi} \) where \( \beta \) is the minimum ordinal such that \( \llbracket\varphi[\mu^{\beta} x. \varphi]\,\vec{\psi}\rrbracket(\rho) = \top \).
  By this process, fixed-point operators annotated with the same sequence of natural numbers have the same ordinal annotation.
  In other words, it defines a function from sequences \( q \) of natural numbers appearing in the trace to ordinal numbers \( \alpha_q \).
  Furthermore one can show that \( \alpha_{q.k} < \alpha_q \) (provided that \( q.k \) appears in the trace).
  Therefore, if the path has a left \( \mu \)-trace witnessed by an infinite sequence \( p \), then \( \alpha_{p[0:1]} > \alpha_{p[0:2]} > \alpha_{p[0:3]} > \dots \) is an infinite decreasing chain of ordinals.

The above argument shows the following theorem.
A detailed proof is in
\iffull \autoref{ap: soundness}.
\else
\cite{KoriCSLfull}.
\fi
\begin{theorem}[Soundness] \label{soundness}
  If there is a cyclic proof %
  of $\Gamma \vdash \Delta$, then $\Gamma \vdash \Delta$ is valid.
\end{theorem}

\section{Completeness of Infinitary Variant}\label{sec: infinitary}
  Our cyclic proof system is incomplete.
  Existence of a proof in our cyclic proof system is computably enumerable, but the valid sequents in \hfl{} are not
  {because \hfl{} includes Peano arithmetic}.

  The aim of this section is to find a complete proof system.
  Following \cite{brotherston}, we study the infinitary variant of our cyclic proof system, in which a proof is an infinite tree instead of a finite tree with cycles.
  More precisely, an \emph{infinitary proof} is a (possibly) infinite derivation tree (without open leaves) of which all infinite paths satisfy the global trace condition.
  Soundness proof of the previous section is applicable to the infinitary system as well.
  Here we discuss its completeness.

\begin{theorem} \label{completeness}
  Let \( \Gamma \vdash \Delta \) be a sequent without higher-order free variables.
  That means, every free variable of the sequent is of type \( \mathbf{N} \) or of type \( \mathbf{N}^{k} \to \mathbf{\Omega} \) for some \( k \geq 0 \).
  If $\Gamma \models \Delta$,
  then $\Gamma \vdash \Delta$ is provable in the infinitary proof system.
\end{theorem}

We give a sketch of the proof.
A detailed proof is given in \autoref{ap: completeness}.
We can assume without loss of generality that the sequent has no free variable of type \( \mathbf{N} \).

  The construction of the infinitary proof is surprisingly straightforward.
  We start from a (finite) pre-proof consisting only of the root node, which is an open leaf, and iteratively expand each open leaf by applying rules ($\lor L/R$), ($\land L/R$), ($\lambda L/R$), and ($\sigma L/R$).
  The only restriction is that the expansion must be ``fair'': the fairness condition we impose is that (i) each open leaf will eventually be expanded, and (ii) for every path, each formula in a sequent will eventually appear in the principal position (unless the path is terminated by the axiom rule).
  This process generates a growing sequence of (finite) pre-proofs, and the infinitary proof is defined as its limit.

  Now it suffices to show the global trace condition of the above constructed candidate proof, but this is the hardest part of the proof.
  Assume an infinite path \( \pi \).
  Since \( \Gamma \models \Delta \), for each valuation \( \rho \), there exists \( \varphi \in \Gamma \) such that \( \llbracket{\varphi}\rrbracket(\rho) = \bot \) or \( \psi \in \Delta \) such that \( \llbracket{\varphi}\rrbracket(\rho) = \top \).
  Then, by a similar argument to the proof of soundness, existence of a left \( \nu \)-trace or a right \( \mu \)-trace leads to a contradiction.
  It is worth noting that the construction is ``dual'': whereas in the soundness proof we ensure \( \llbracket \tau_i \rrbracket(\rho) = \top \) for a left-trace \( (\tau_i)_{i \geq 0} \), here \( \llbracket \tau_i \rrbracket(\rho) = \bot \) is kept.
  This difference allows us to construct a trace of the intended path \( \pi \).
  By appropriately choosing \( \rho \), one can ensure that the generated trace is infinite; hence we have constructed an infinite trace that is not left \( \nu \)- nor right \( \mu \)-trace.
  The proof is completed by the following lemma, which is technical but often used in the analysis of HFL (cf.~\cite[Lemma {6 \& 7}]{bruse14}, \cite[Lemma 26, Appendix E.2]{kobayashi18} and \cite[Lemma 14]{tsukada20}):
\begin{lemma} \label{lem: uniqueness}
  Every infinite trace $(\tau_i)_{i \leq 0}$ is either a $\mu$-trace or a $\nu$-trace but not both.
\end{lemma}

\begin{remark}
    We are not sure if the proof can be extended to sequents with higher-order free variables.
    The problem is the construction of \( \rho \).
    In the current setting, for each free variable \( f \) of type \( \mathbf{N}^{k} \to \mathbf{\Omega} \), the valuation \( \rho \) is defined for follows: for each \( k \)-tuple \( \vec{m} \) of natural numbers,
    (i) if \( f\,\vec{m} \) appears on the left side of a sequent in the path, then \( \rho(f)(\vec{m}) := \top \);
    (ii) if \( f\,\vec{m} \) appears on the right side of a sequent in the path, then \( \rho(f)(\vec{m}) := \bot \); and
    (iii) otherwise the value of \( \rho(f)(\vec{m}) \) is arbitrary.
    The point is that the values of arguments of each fully-applied occurrence of \( f \) is canonically determined.
    This construction of \( \rho \) is no longer possible in the presence of a higher-order free variable, say \( g : \mathbf{\Omega} \to \mathbf{\Omega} \).
    When \( g\,(f\,\vec{m}) \) appears on the left side, there are two assignments that make this formula true, namely \( \rho_1(g)(\bot) = \top \) with \( \rho_1(f)(\vec{m}) = \bot \) and \( \rho_2(g)(\top) = \top \) with \( \rho_2(f)(\vec{m}) = \top \).
    \qed
\end{remark}

\setcounter{section}{6}
\section{Related Work}
\subsection{Cyclic Proof Systems}

The idea of cyclic proof can at least be traced back to the work of
Sprenger and Dam~\cite{sprenger} on a cyclic proof system called $S_{glob}$,
where proofs are explicitly annotated with ordinals.
The ordinal annotations are not convenient for automated theorem proving.

Brotherston and Simpson~\cite{brotherston}
proposed a cyclic proof system called CLKID for
a first-order predicate logic with inductive definitions,
which does not require ordinal annotations.
They have proved soundness of CLKID,
and also shown that a proof system with explicit induction rules (called LKID)
can be embedded into CLKID; we have shown analogous results in
Theorem~\ref{soundness} and Proposition~\ref{prop: explicit} for \hfl{}.
Doumane~\cite{doumane} formalized a cyclic proof system for the linear-time
(propositional) \(\mu\)-calculus, which features alternating fixed-points.
Our global trace condition has been inspired by her trace condition.

Besides cyclic proof systems for ordinary first-order logics,
cyclic proof systems have also been proposed
for separation logics~\cite{brotherston11, brotherston12, tellez20},
and automated tools have been developed based on those cyclic proof systems.

\subsection{Proof Systems and Games for HFL}

HFL was originally proposed by Viswanathan and Viswanathan~\cite{viswanathan04},
and its extensions with arithmetic (such as \hfl{}) have
recently been drawing attention in the context of higher-order program
verification~\cite{kobayashi18,watanabe19}.

For pure HFL (without natural numbers), Kobayashi et al.~\cite{kobayashi17}
proposed a type-based inference system for proving the validity of HFL formulas.%
\footnote{Actually, they formalized a type system for model checking,
  which can also be used as a proof system for proving validity of \hfl{} formulas
  without natural numbers.}
It is left for future work to study the relationship between their type system and
our cyclic proof system; it would be interesting to see whether their type system
can be embedded in our cyclic proof system.

Tsukada~\cite{tsukada20} gave a game-based characterization of
\hfl. His game may be considered a special case of the infinitary
version of our proof system,
where formulas are restricted to those
whose free variables have only natural number types.

Burn et al.~\cite{hochc} proposed a refinement type system for HoCHC, which
is closely related to the \(\nu\)-only fragment of \hfl{}.
Their proof system is incomplete, and does not support fixed-point alternations.

\section{Conclusion}
We have %
proposed a cyclic proof system for \hfl,
a higher-order logic with natural numbers and alternating fixed-points,
which we expect to be useful for higher-order program verification.
Our proof system has been inspired
by previous cyclic proof systems for first-order logics~\cite{brotherston, doumane}.
We have shown the soundness of the proof system and the decidability of
the global trace condition. We have also shown
a restricted form of standard completeness for the infinitary version of our proof system.

Constructing a (semi-)automated tool based on our cyclic proof system
is left for future work. On the theoretical side, we plan to study
whether Henkin completeness and
the cut elimination property hold
for (possibly a variation of) our cyclic proof system.

\bibliography{csl21}

\clearpage
\appendix
\section{Definition of relevant occurrences} \label{ap: rel-occurrence}
{
Here we give a detailed definition of the notion of \emph{relevant occurrences}
introduced after Definition~\ref{def: preproof}.
We have already defined relevant occurrences
for the rules (\(\wedge L\)), (\(\vee L\)), and  (Mono).
We give the definition for the other rules. Below, \(n_i\) refers to the
conclusion of each rule, and \(n_{i+1}\) refers to one of the premises.
In all the rules, each formula in \(\Gamma\) or \(\Delta\) in \(n_{i+1}\)
is a relevant occurrence of the corresponding formula in \(\Gamma\) or \(\Delta\)
in \(n_i\).
\begin{itemize}
\item In (Cut), \(\varphi\) in \(n_{i+1}\) is not relevant to any formula in \(n_i\).
\item In (Ctr L) and (Ctl R), both occurrences of \(\varphi\) are relevant to
  \(\varphi\) in \(n_i\).
\item In (Ex L) and (Ex R), \(\psi\) (\(\varphi\), resp.) in \(n_{i+1}\) is
  a relevant occurrence of \(\psi\) (\(\varphi\), resp.) in \(n_{i}\).
\item In (Subst), each formula \(\psi\) in \(\Gamma\) (\(\Delta\), resp.) of $n_{i+1}$ is
  relevant to \(\psi[\varphi/x]\) in \(\Gamma[\varphi/x]\) (\(\Delta[\varphi/x]\), resp.) of $n_{i}$.
\item In ($=$L), each formula \(\psi[t/x, s/y]\) in $\Gamma[t/x, s/y]$ ($\Delta[t/x, s/y]$, resp.)
   of $n_{i+1}$
  is relevant to \(\psi[s/x, t/y]\) in $\Gamma[s/x, t/y]$ ($\Delta[s/x, t/y]$, resp.) of $n_{i}$.
\item In ($\lor$R), \(\varphi\) and \(\psi\) in \(n_{i+1}\) are
  relevant to \(\varphi\lor\psi\) in \(n_i\).
\item In ($\land$R), \(\varphi\) and \(\psi\) in \(n_{i+1}\) are
  relevant to \(\varphi\land\psi\) in \(n_i\).
\item In ($\lambda$L) and ($\lambda$R), $\varphi[\psi/x] \, \vec{\psi}$ in $n_{i+1}$ is relevant to $(\lambda x. \varphi) \, \psi \, \vec{\psi}$ in $n_i$.
\item In ($\sigma$L) and ($\sigma$R), $\varphi[\sigma x. \varphi/x] \, \vec{\psi}$ in $n_{i+1}$ is relevant to $(\sigma x. \varphi) \, \vec{\psi}$ in $n_i$.
\item In (Nat), $N \, x$ in $n_{i+1}$ is not relevant to any formula in $n_i$.
\item In (P2), $s=t$ in $n_{i+1}$ is relevant to $\mathbf{S}s = \mathbf{S}t$ in $n_i$.
\end{itemize}
}

\iffull
\section{Proof of \autoref{prop: quantifiers}} \label{ap: quantifiers}

\begin{lemma} \label{ap: z leq t}
  There exists a cyclic proof of $\vdash \mathbf{Z} \leq t$
  where $\leq \ \equiv (\mu Y. \lambda n. \lambda m. (n = m) \lor Y \, (\mathbf{S}n, m))$.
\end{lemma}
\begin{proof}
  The proof is as below:
\begin{center}
  \AxiomC{}
  \UnaryInfC{$\vdash \mathbf{Z} = \mathbf{Z}$}
  \RightLabel{($\mu$R, $\lor$R, Wk)}
  \UnaryInfC{$\vdash \mathbf{Z} \leq \mathbf{Z}$}
  \RightLabel{($=$L)}
  \UnaryInfC{$t=\mathbf{Z} \vdash \mathbf{Z} \leq t$}
    \AxiomC{$(\star) \ N \, t \vdash \mathbf{Z} \leq t$}
    \UnaryInfC{$N \, t' \vdash \mathbf{Z} \leq t'$}
      \AxiomC{$\Pi$}
      \RightLabel{(Cut)}
    \BinaryInfC{$N \, t' \vdash \mathbf{SZ} \leq \mathbf{S}t'$}
    \RightLabel{($\mu$R, $\lor$R, Wk)}
    \UnaryInfC{$N \, t' \vdash \mathbf{Z} \leq \mathbf{S}t'$}
    \RightLabel{(=L)}
    \UnaryInfC{$N \, t', t=\mathbf{S}t' \vdash \mathbf{Z} \leq t$}
    \RightLabel{($\exists$L, $\land$L)}
    \UnaryInfC{$\exists t'. N \, t' \land t=\mathbf{S}t' \vdash \mathbf{Z} \leq t$}
    \RightLabel{($\mu$L)}
  \BinaryInfC{$(\star) \ N \, t \vdash \mathbf{Z} \leq t$}
  \RightLabel{(Nat)}
  \UnaryInfC{$\vdash \mathbf{Z} \leq t$}
  \DisplayProof
\end{center}
\begin{center}
  $\Pi := $
  \AxiomC{}
  \RightLabel{(=R)}
  \UnaryInfC{$\vdash \mathbf{S}t' = \mathbf{S}t'$}
  \RightLabel{($\mu$R, $\lor$R)}
  \UnaryInfC{$\vdash \mathbf{S}t' \leq \mathbf{S}t'$}
  \RightLabel{(=L)}
  \UnaryInfC{$z = t' \vdash \mathbf{S}z \leq \mathbf{S}t'$}
    \AxiomC{$(\dagger) \ z \leq t' \vdash \mathbf{S}z \leq \mathbf{S}t'$}
    \RightLabel{(Subst)}
    \UnaryInfC{$\mathbf{S}z \leq t' \vdash \mathbf{SS}z \leq \mathbf{S}t'$}
    \RightLabel{(Wk R)}
    \UnaryInfC{$\mathbf{S}z \leq t' \vdash \mathbf{S}z = \mathbf{S}t', \mathbf{SS}z \leq \mathbf{S}t'$}
    \RightLabel{($\mu$R, $\lor$R)}
    \UnaryInfC{$\mathbf{S}z \leq t' \vdash \mathbf{S}z \leq \mathbf{S}t'$}
  \RightLabel{($\mu$L)}
  \BinaryInfC{$(\dagger) \ z \leq t' \vdash \mathbf{S}z \leq \mathbf{S}t'$}
  \RightLabel{(Wk, Subst)}
  \UnaryInfC{$N \, t', \mathbf{Z} \leq t' \vdash \mathbf{SZ} \leq \mathbf{S}t'$}
  \DisplayProof
\end{center}
For each infinite path $\pi$ in the above pre-proof,
if $\pi$ goes along $(\dagger)$ infinitely often then we can trace the left $\mu$ in the predicate $\leq$;
otherwise, we can trace the left $\mu$ in $N$.
Thus this pre-proof satisfies the global trace condition.
\end{proof}

\begin{proof}[Proof (\autoref{prop: quantifiers}).]
  We first replace quantifier rules by the following proofs and construct a pre-proof $\Pi'$ with no quantifier rules.

  \begin{itemize}
    \item If the rule is $(\exists L)$ and $x^N$:

      \begin{center}
      \AxiomC{$\Gamma, \varphi \vdash \Delta$}
        \AxiomC{$(\star) \ \Gamma, (\mu E. \lambda y. \varphi[y/x] \lor E \, (\mathbf{S}y)) \, x \vdash \Delta$}
        \RightLabel{(Subst) since $x \not \in FV(\Gamma, \Delta)$}
        \UnaryInfC{$\Gamma, (\mu E. \lambda y. \varphi[y/x] \lor E \, (\mathbf{S}y)) \, (\mathbf{S}x) \vdash \Delta$}
      \RightLabel{($\lor L$)}
      \BinaryInfC{$\Gamma, \varphi \lor (\mu E. \lambda y. \varphi[y/x] \lor E \, (\mathbf{S}y)) \, (\mathbf{S}x) \vdash \Delta$}
      \RightLabel{($\mu L$)}
      \UnaryInfC{$(\star) \ \Gamma, (\mu E. \lambda y. \varphi[y/x] \lor E \, (\mathbf{S}y)) \, x \vdash \Delta$}
      \RightLabel{(Subst)}
      \UnaryInfC{$\Gamma, (\mu E. \lambda y. \varphi[y/x] \lor E \, (\mathbf{S}y)) \, \mathbf{Z} \vdash \Delta$}
      \DisplayProof
      \end{center}

    \item If the rule is $(\exists L)$ and $x^T$:

      \begin{center}
      \AxiomC{$\Gamma, \varphi \vdash \Delta$}
      \RightLabel{(Subst) since $x \not \in FV(\Gamma, \Delta)$}
      \UnaryInfC{$\Gamma, \varphi[\top_T/x] \vdash \Delta$}
      \RightLabel{($\lambda$L)}
      \UnaryInfC{$\Gamma, (\lambda x. \varphi) \, \top_T \vdash \Delta$}
      \DisplayProof
      \end{center}

    \item If the rule is $(\exists R)$ and $x^N$:

      \begin{center}
        \AxiomC{$\Gamma \vdash \varphi[t/x], \Delta$}
        \RightLabel{($\mu R, \lor R$)}
        \UnaryInfC{$\Gamma \vdash E_\varphi \, t, \Delta$}
        \RightLabel{($=L$)}
        \UnaryInfC{$\Gamma, z = t \vdash E_\varphi \, z, \Delta$}
          \AxiomC{$(\star) \ \Gamma, z \leq t \vdash E_\varphi \, z, \Delta$}
          \RightLabel{(Subst)}
          \UnaryInfC{$\Gamma, \mathbf{S}z \leq t \vdash E_\varphi \, (\mathbf{S}z), \Delta$}
          \RightLabel{($\mu R, \lor R$)}
          \UnaryInfC{$\Gamma, \mathbf{S}z \leq t \vdash E_\varphi \, z, \Delta$}
          \RightLabel{($\lor L$)}
        \BinaryInfC{$\Gamma, z = t \lor \mathbf{S}z \leq t \vdash E_\varphi \, z, \Delta$}
        \RightLabel{($\mu L$)}
        \UnaryInfC{$(\star) \ \Gamma, z \leq t \vdash E_\varphi \, z, \Delta$}
        \RightLabel{(Subst) where $z$ is a free variable}
        \UnaryInfC{$\Gamma, \mathbf{Z} \leq t \vdash E_\varphi \, \mathbf{Z}, \Delta$}
        \RightLabel{(Cut) and $\vdash \mathbf{Z} \leq t$ (see \autoref{ap: z leq t})}
        \UnaryInfC{$\Gamma \vdash E_\varphi \, \mathbf{Z}, \Delta$}
        \DisplayProof
      \end{center}
      where $E_\varphi \equiv (\mu E. \lambda y. \varphi[y/x] \lor E \, (\mathbf{S}y))$
      and $\leq \ \equiv (\mu Y. \lambda n. \lambda m. n = m \lor Y \, (\mathbf{S}n, m))$.

    \item If the rule is $(\exists R)$ and $x^T$:

      \begin{center}
        \AxiomC{$\Gamma \vdash \varphi[\psi/x], \Delta$}
        \RightLabel{(Wk R)}
        \UnaryInfC{$\Gamma \vdash \varphi[\psi/x], \varphi[\top_T/x], \Delta$}
          \AxiomC{$(\star) \ \Gamma, \psi \, \vec{y} \vdash \top_T \, \vec{y}, \Delta$}
          \RightLabel{($\nu$R)}
          \UnaryInfC{$(\star) \ \Gamma, \psi \, \vec{y} \vdash \top_T \, \vec{y}, \Delta$}
          \RightLabel{(Mono)}
          \UnaryInfC{$\Gamma, \varphi[\psi/x] \vdash \varphi[\top_T/x], \Delta$}
        \RightLabel{(Cut)}
        \BinaryInfC{$\Gamma \vdash \varphi[\top_T/x], \Delta$}
        \RightLabel{($\lambda$R)}
        \UnaryInfC{$\Gamma \vdash (\lambda x. \varphi) \, \top_T, \Delta$}
        \DisplayProof
      \end{center}
    \item If the rule is $(\forall L)$ or $(\forall R)$ then the same method can be applied as $(\exists R)$ or $(\exists L)$.
  \end{itemize}
We need to check whether the pre-proof $\Pi'$ satisfies the global trace condition.
For all infinite paths $\pi$,
if there exits a tail of $\pi$ which goes along only one cycle made by the above construction, then we choose the trace of the cycle.
Otherwise, the corresponding path in $\Pi$ is infinite and this path satisfies the global trace condition.
Therefore we can choose a left $\mu$-trace or right $\nu$-trace in $\pi$.
\end{proof}

\section{Proof of \autoref{soundness}} \label{ap: soundness}
  Let \( f \) be a function on a complete lattice.
  For each ordinal \( \alpha \), we define \( f^{\alpha}(\bot) \) by:
  \begin{align*}
    f^{0}(\bot) = \bot
    \quad\mbox{and}\quad
    f^{\alpha}(\bot) = \bigsqcup_{\beta < \alpha} f(f^{\beta}(\bot)).
  \end{align*}
  Similarly \( f^{\alpha}(\top) \) is defined by:
  \begin{align*}
    f^{0}(\top) = \top
    \quad\mbox{and}\quad
    f^{\alpha}(\top) = \bigsqcap_{\beta < \alpha} f(f^{\beta}(\top)).
  \end{align*}
\begin{lemma}[Cousot-Cousot~\cite{cousot79}] \label{lem: cousot}
Let $(C, \leq)$ be a complete lattice, $f$ be a monotone function
and $\gamma$ be an ordinal number greater than the cardinality of $C$.
Then $f^{\gamma}(\bot)$ is the least fixed-point of $f$ and $f^{\gamma}(\top)$ is the greatest fixed-point of $f$.
\end{lemma}

We extend the definition of formulas by allowing $\mu$ and $\nu$ to have ordinal numbers like $\mu^\alpha, \nu^\alpha$.
The definition of $\llbracket . \rrbracket$ for $\mu^\alpha, \nu^\alpha$ is as follows.
\begin{align*}
\llbracket \mathcal{H} \vdash (\mu^\alpha x^T. \varphi) \, \vec{\psi} \rrbracket(\rho) &= ((\lambda v. \llbracket \mathcal{H}, x: T \vdash \varphi \rrbracket(\rho[x  \mapsto v]))^{\alpha} \, (\bot_n)) \, \rho(\vec{\psi}) \\
\llbracket \mathcal{H} \vdash (\nu^\alpha x^T. \varphi) \, \vec{\psi} \rrbracket(\rho) &= ((\lambda v. \llbracket \mathcal{H}, x: T \vdash \varphi \rrbracket(\rho[x  \mapsto v]))^{\alpha} \, (\top_n)) \, \rho(\vec{\psi})
\end{align*}

\begin{proof}[Proof (\autoref{soundness}).]
Assume there exists a cyclic proof $\Pi$ of $\Gamma \vdash \Delta$.
We can naturally transform this cyclic proof to an infinitary proof $\Pi$ without (Subst).
We assume towards a contradiction that the root sequent $\Gamma \vdash \Delta$ is invalid.
Let $\rho_0$ be a valuation such that $\Gamma \not \models_{\rho_0} \Delta$.

We will define
an infinite path $(n_i)_{i \geq 0} = (\Gamma_i \vdash \Delta_i)_{i \geq 0}$ starting from $\Gamma \vdash \Delta$ in $\Pi$
, sequents $(n_i')_{i \geq 0} = (\Gamma'_i \vdash \Delta'_i)_{i \geq 0}$
and valuations $(\rho_i)_{i \geq 0}$ which satisfy the following four conditions:
\begin{enumerate}
\item $\Gamma_i \vdash \Delta_i$ is obtained by forgetting ordinal numbers of $\Gamma'_i \vdash \Delta'_i$;
\item each $\mu$ in $\Gamma_i'$ (resp.~$\nu$ in $\Delta_i$) has an ordinal number;
\item $\Gamma_i' \not \models_{\rho_i} \Delta_i'$ holds;
\item the ordinal number is determined by the sequence of natural numbers generated by Definition~\ref{def: trace-ho}; and
\item for each finite sequence \( q \) and natural number \( k \), the ordinal for \( q.k \) is smaller than that for \( q \).
\end{enumerate}

The definition of $(n_i, n_i', \rho_i)_{i \geq 0}$ is as below:

\begin{itemize}
  \item $n_0 := \Gamma \vdash \Delta$.
        Because $\Gamma_0 \vdash \Delta_0$ is invalid under $\rho_0$,
        we can get $\Gamma'_0 \vdash \Delta'_0$ which is invalid under $\rho_0$ by assigning ordinal numbers to all $\mu$ in $\Gamma$ and $\nu$ in $\Delta$.
    \item If $n_i, n_i', \rho_i$ are defined, then we define $n_{i+1}, n_{i+1}', \rho_{i+1}$ as follows.
  \begin{itemize}
    \item If $n_i$ is an axiom, $\Gamma_i \vdash \Delta_i$ must be valid and then $\Gamma_i' \vdash \Delta_i'$ must be so too. This is a contradiction.
    \item If $n_i$ is an open leaf:
      $n_{i+1} := \mathcal{R}(n_i), \Gamma'_{i+1} := \Gamma'_i, \Delta'_{i+1} := \Delta'_i, \rho_{i+1} := \rho_i$.
    \item Otherwise, we devide cases by the rule which has $n_i$ as the conclusion.
    Let $p_0, p_1, ...$ be the premises sequents from left to right.
    Below, $\varphi'$ means the formula in $\Gamma_i' \cup \Delta_i'$ which corresponds to $\varphi$ in $\Gamma_i \cup \Delta_i$.
    \begin{itemize}
      \item Case (Cut):
        \begin{center}
          \AxiomC{$\Gamma \vdash \varphi, \Delta$}
          \AxiomC{$\Gamma, \varphi \vdash \Delta$}
          \RightLabel{(Cut)}
          \BinaryInfC{$\Gamma \vdash \Delta$}
          \DisplayProof
        \end{center}
        If there exists a valuation $\rho$ such that
        (1) $\rho$ is equal to $\rho_i$ in $FV(\Gamma, \Delta)$ (i.e.\ $x \in FV(\Gamma, \Delta) \Rightarrow \rho(x) = \rho_i(x)$) and
        (2) $\llbracket \varphi \rrbracket_\rho = \top$, then $n_{i+1} := p_1, \rho_{i+1} := \rho$.
        Since $\llbracket \varphi \rrbracket_\rho = \top$, there exists $\varphi'$ which holds $\llbracket \varphi' \rrbracket_{\rho} = \top$ by assigning ordinal numbers to all $\mu$ in $\varphi$.
        Then $\Gamma_{i+1}' := (\Gamma_i', \varphi'), \Delta'_{i+1} := \Delta'_i$.

        Otherwise, there exists $\rho$ such that (1) and (3) $\llbracket \varphi \rrbracket_{\rho} = \bot$.
        By using this $\rho$, we execute $n_{i+1} := p_0, \rho_{i+1} := \rho$.
        Since $\llbracket \varphi \rrbracket_{\rho} = \bot$, there exists $\varphi'$ which holds $\llbracket \varphi' \rrbracket_{\rho} = \bot$ by assigning ordinal numbers to all $\nu$ in $\varphi$.
        $\Gamma_{i+1}' := \Gamma_i', \Delta_{i+1}' := (\varphi', \Delta_i')$.

         \item Case (Mono):
           \begin{center}
            \AxiomC{$\{\Gamma, \psi \, \vec{y} \vdash \chi \, \vec{y}, \Delta\}_{1, \cdots, k}$}
            \RightLabel{$\vec{y} \cap FV(\Gamma, \psi, \chi, \Delta) = \emptyset$, (Mono)}
            \UnaryInfC{$\Gamma, \varphi[\psi/x^T] \vdash \varphi[\chi/x^T], \Delta$}
            \DisplayProof
           \end{center}
           Let $\psi_j$ and $\chi_j$ denote $\psi$ and $\chi$ in the $j$-th premise.
           Suppose $\llbracket \psi'_j \rrbracket_{\rho_i} \leq \llbracket \chi'_j \rrbracket_{\rho_i}$ for all $j \in \{1, \cdots, k\}$
           then $\llbracket (\varphi[\psi/x])' \rrbracket_{\rho_i} \leq \llbracket (\varphi[\chi/x])' \rrbracket_{\rho_i}$
           because $\llbracket \lambda x. \varphi' \rrbracket_{\rho_i}$ is a monotone function.
           This contradicts $\Gamma_i' \not \models_{\rho_i} \Delta_i'$
           thus there exists $j \in \{1, \cdots, k\}$ and
           $\vec{v} \in \llbracket T \rrbracket$ such that
           $\llbracket \psi'_j \rrbracket_{\rho_i}\vec{v} = \top$ and $\llbracket \chi'_j \rrbracket_{\rho_i}\vec{v} = \bot$.
           $n_{i+1} := p_{j-1}, \rho_{i+1} := \rho_i[\vec{y} \mapsto \vec{v}], \Gamma_{i+1}' := (\Gamma', \psi'_j \, \vec{y}), \Delta_{i+1}' := (\chi'_j \, \vec{y}, \Delta')$.

         \item Case ($\lor L$):
          \begin{center}
            \AxiomC{$\Gamma, \varphi \vdash \Delta$}
            \AxiomC{$\Gamma, \psi \vdash \Delta$}
            \RightLabel{($\lor L$)}
            \BinaryInfC{$\Gamma, \varphi \lor \psi\vdash \Delta$}
            \DisplayProof
          \end{center}
          $\llbracket \varphi' \lor \psi' \rrbracket_{\rho_i}
          = \llbracket \varphi' \rrbracket_{\rho_i} \lor \llbracket \psi' \rrbracket_{\rho_i}
          = \top$
          since $\Gamma', \varphi' \lor \psi' \not \models_{\rho_i} \Delta'$.
          Therefore, $\llbracket \varphi' \rrbracket_{\rho_i} = \top$ or $\llbracket \psi' \rrbracket_{\rho_i}$.
          If $\llbracket \varphi' \rrbracket_{\rho_i} = \top$ then $n_{i+1} := p_0, \rho_{i+1} := \rho_i, \Gamma_{i+1}' := (\Gamma', \varphi'), \Delta_{i+1}' := \Delta_i'$ else
          $n_{i+1} := p_1, \rho_{i+1} := \rho_i, \Gamma_{i+1}' := (\Gamma', \psi'), \Delta'_{i+1} := \Delta_i'$.

        \item Case ($\land R$):
          \begin{center}
            \AxiomC{$\Gamma \vdash \varphi, \Delta$}
            \AxiomC{$\Gamma \vdash \psi, \Delta$}
            \RightLabel{($\land R$)}
            \BinaryInfC{$\Gamma \vdash \varphi \land \psi, \Delta$}
            \DisplayProof
          \end{center}
          $\llbracket \varphi' \land \psi' \rrbracket_{\rho_i}
          = \llbracket \varphi' \rrbracket_{\rho_i} \land \llbracket \psi' \rrbracket_{\rho_i}
          = \bot$
          since $\Gamma' \not \models_{\rho_i} \varphi' \land \psi', \Delta'$.
          Therefore, $\llbracket \varphi' \rrbracket_{\rho_i} = \bot$ or $\llbracket \psi' \rrbracket_{\rho_i} = \bot$.
          If $\llbracket \varphi' \rrbracket_{\rho_i} = \bot$ then $n_{i+1} := p_0, \rho_{i+1} := \rho_i, \Gamma_{i+1}' := \Gamma_i', \Delta_{i+1}' := (\varphi', \Delta')$ else $n_{i+1} := p_1, \rho_{i+1} := \rho_i, \Gamma_{i+1}' := \Gamma_i', \Delta_{i+1}' := (\psi', \Delta')$.

        \item Case ($\mu L$):
          \begin{center}
            \AxiomC{$\Gamma, \varphi[\mu x. \varphi/x] \, \vec{\psi} \vdash \Delta$}
            \RightLabel{($\mu L$)}
            \UnaryInfC{$\Gamma, (\mu x. \varphi) \, \vec{\psi} \vdash \Delta$}
            \DisplayProof
          \end{center}

          Let $\alpha$ be the ordinal number which assigned to the head $\mu$ of $(\mu x. \varphi) \, \vec{\psi}$.
          If $\alpha = 0$ then
          $\llbracket (\mu^0 x^T. \varphi') \, \vec{\psi'} \rrbracket_{\rho_i}
          = \llbracket (\bot_T) \, \vec{\psi'} \rrbracket_{\rho_i} = \bot$
          and this contradicts the assumption.
          \begin{align*}
          \llbracket (\mu^\alpha x^T. \varphi') \, \vec{\psi'} \rrbracket_{\rho_i}
          &= ((\lambda  v. \llbracket \varphi' \rrbracket_{\rho_i[x \mapsto v]})^{\alpha} \, (\bot_T)) \, (\llbracket \vec{\psi'} \rrbracket_{\rho_i}) \\
          &= (\bigsqcup_{\beta < \alpha} (\lambda v. \llbracket \varphi' \rrbracket_{\rho_i[x \mapsto v]}) \, ((\lambda v. \llbracket \varphi' \rrbracket_{\rho_i[x \mapsto v]})^{\beta}  \, (\bot_T))) \, (\llbracket \vec{\psi'} \rrbracket_{\rho_i}) \\
          &= \top
          \end{align*}
          Thus there exists $\beta < \alpha$ such that
          \begin{center}
            $((\lambda v. \llbracket \varphi' \rrbracket_{\rho_i[x \mapsto v]}) \, ((\lambda v. \llbracket \varphi' \rrbracket_{\rho_i[x \mapsto v]})^{\beta}  \, (\bot_T))) \, (\llbracket \vec{\psi'} \rrbracket_{\rho_i}) = \top$.
          \end{center}

          $n_{i+1} := p_0, \rho_{i+1} := \rho_i, \Gamma'_{i+1} := (\Gamma', \varphi'[\mu^\beta x. \varphi'/x] \, \vec{\psi'}), \Delta'_{i+1} := \Delta'_i$ then
          \begin{align*}
          \llbracket (\varphi'[\mu^\beta x. \varphi'/x]) \, \vec{\psi'} \rrbracket_{\rho_{i}}
          &= ((\lambda  v. \llbracket \varphi' \rrbracket_{\rho_{i}[x \mapsto v]}) \, (\llbracket \mu^\beta x. \varphi' \rrbracket_{\rho_{i}})) \, (\llbracket \vec{\psi'} \rrbracket_{\rho_i}) \\
          &= ((\lambda  v. \llbracket \varphi' \rrbracket_{\rho_{i}[x \mapsto v]}) \, ((\lambda v. \llbracket \varphi' \rrbracket_{\rho_{i}[x \mapsto v]})^{\beta} \, (\bot_T))) \, (\llbracket \vec{\psi'} \rrbracket_{\rho_i}) \\
          &= \top
          \end{align*}

        \item Case ($\nu R$):
          \begin{center}
            \AxiomC{$\Gamma \vdash \varphi[\nu x. \varphi/x] \, \vec{\psi}, \Delta$}
            \RightLabel{($\nu R$)}
            \UnaryInfC{$\Gamma \vdash (\nu x. \varphi) \, \vec{\psi}, \Delta$}
            \DisplayProof
          \end{center}

          Let $\alpha$ be the ordinal number which assigned to the head $\nu$ of $(\nu x. \varphi) \, \vec{\psi}$.
          If $\alpha = 0$ then
          $\llbracket (\nu^0 x^T. \varphi') \, \vec{\psi'} \rrbracket_{\rho_i}
          = \llbracket (\top_T) \, \vec{\psi'} \rrbracket_{\rho_i} = \top$
          and this contradicts the assumption.
          \begin{align*}
          \llbracket (\nu^\alpha x^T. \varphi') \, \vec{\psi'} \rrbracket_{\rho_i}
          &= ((\lambda  v. \llbracket \varphi' \rrbracket_{\rho_i[x \mapsto v]})^{\alpha} \, (\bot_T)) \, (\llbracket \vec{\psi'} \rrbracket_{\rho_i}) \\
          &= (\bigsqcap_{\beta < \alpha} (\lambda v. \llbracket \varphi' \rrbracket_{\rho_i[x \mapsto v]}) \, ((\lambda v. \llbracket \varphi' \rrbracket_{\rho_i[x \mapsto v]})^{\beta}  \, (\top_T))) \, (\llbracket \vec{\psi'} \rrbracket_{\rho_i}) \\
          &= \bot
          \end{align*}
          Thus there exists $\beta < \alpha$ such that
          \begin{center}
          $((\lambda v. \llbracket \varphi' \rrbracket_{\rho_i[x \mapsto v]}) \, ((\lambda v. \llbracket \varphi' \rrbracket_{\rho_i[x \mapsto v]})^{\beta} \,  (\top_T))) \, (\llbracket \vec{\psi'} \rrbracket_{\rho_i}) = \bot$.
          \end{center}

          $n_{i+1} := p_0, \rho_{i+1} := \rho_i, \Gamma_{i+1}' := \Gamma_i', \Delta_{i+1}' := (\varphi'[\nu^\beta x. \varphi'/x] \, \vec{\psi'}, \Delta')$ then
          \begin{align*}
          \llbracket (\varphi'[\nu^\beta x. \varphi'/x]) \, \vec{\psi'} \rrbracket_{\rho_{i}}
          &= ((\lambda  v. \llbracket \varphi' \rrbracket_{\rho_{i}[x \mapsto v]}) \, (\llbracket \nu^\beta x. \varphi' \rrbracket_{\rho_{i}})) \, (\llbracket \vec{\psi'} \rrbracket_{\rho_i}) \\
          &= ((\lambda  v. \llbracket \varphi' \rrbracket_{\rho_{i}[x \mapsto v]}) \, ((\lambda v. \llbracket \varphi' \rrbracket_{\rho_{i}[x \mapsto v]})^{\beta} \, (\top_T))) \, (\llbracket \vec{\psi'} \rrbracket_{\rho_i}) \\
          &= \bot
          \end{align*}

      \item Case (Nat):
        \begin{center}
          \AxiomC{$\Gamma, N \, x^{\mathbf{N}} \vdash \Delta$}
          \RightLabel{(Nat)}
          \UnaryInfC{$\Gamma \vdash \Delta$}
          \DisplayProof
        \end{center}
        If $x \in FV(\Gamma, \Delta)$ then $\rho_{i+1} := \rho_i$ and otherwise $\rho_{i+1} := \rho_i[x \mapsto \mathbf{Z}]$.
        For each case, there is an ordinal number $\alpha$
        such that $\llbracket N^\alpha \, x \rrbracket_{\rho_{i+1}} = \top$.
        Then $n_{i+1} := p_0, \Gamma_{i+1}' := (\Gamma_i', N^\alpha \, x), \Delta_{i+1}' := \Delta_i'$.

        \item Otherwise: $n_{i+1} := p_0, \rho_{i+1} := \rho_i, \Gamma_{i+1}' := \Gamma_i', \Delta_{i+1}' := \Delta'_i$.
    \end{itemize}
  \end{itemize}
\end{itemize}

Because of the global trace condition, $(\Gamma_i \vdash \Delta_i)_{i \geq 0}$ has a trace $(\tau_i)_{i \geq k}$ which is a left $\mu$-trace or a right $\nu$-trace.
Then there is an infinite sequence $p$ of the trace by the definition of $\sigma$-trace.
For all $n \geq 1$, there is a $\sigma_{p[0:n]}$ in some $\tau_i$
and we denote by $\alpha_n$ the ordinal number assigned to this $\sigma$.

By the definiton of $(\alpha_n)_{n \geq 1}$ and $(\Gamma'_i \vdash \Delta'_i)_{i \geq 0}$,
every ordinal number assigned to $\sigma_{p[0:n]}$ is $\alpha_n$.
Thus for all $n \geq 1$, there exists $i$ such that $(\tau_i \equiv (\sigma_{p[0:n]} x. \varphi) \, \vec{\psi}) \to (\tau_{i+1} \equiv (\varphi[\sigma_{p[0:n+1]} x. \varphi/x]) \, \vec{\psi})$.
Hence $\alpha_n > \alpha_{n+1}$ holds for all $n \geq 1$ so $(\alpha_n)_{n \geq 1}$ is a decreasing sequence of ordinal numbers.
However, such a sequence does not exist so this is a contradiction.
\end{proof}
\else
\fi

\section{Proof of \autoref{completeness}} \label{ap: completeness}
\iffull
\else
  Let \( f \) be a function on a complete lattice.
  For each ordinal \( \alpha \), we define \( f^{\alpha}(\bot) \) by
  \( f^{0}(\bot) = \bot \) and \( f^{\alpha}(\bot) = \bigsqcup_{\beta < \alpha} f(f^{\beta}(\bot)) \).
  Similarly \( f^{\alpha}(\top) \) is defined by
    $f^{0}(\top) = \top$ and
    $f^{\alpha}(\top) = \bigsqcap_{\beta < \alpha} f(f^{\beta}(\top))$.
\begin{lemma}[Cousot-Cousot~\cite{cousot79}] \label{lem: cousot}
Let $(C, \leq)$ be a complete lattice, $f$ be a monotone function
and $\gamma$ be an ordinal number greater than the cardinality of $C$.
Then $f^{\gamma}(\bot)$ is the least fixed-point of $f$ and $f^{\gamma}(\top)$ is the greatest fixed-point of $f$.
\end{lemma}

We extend the definition of formulas by allowing $\mu$ and $\nu$ to have ordinal numbers like $\mu^\alpha, \nu^\alpha$.
The definition of $\llbracket . \rrbracket$ for $\mu^\alpha, \nu^\alpha$ is as follows.
\begin{align*}
\llbracket \mathcal{H} \vdash (\mu^\alpha x^T. \varphi) \, \vec{\psi} \rrbracket(\rho) &= ((\lambda v. \llbracket \mathcal{H}, x: T \vdash \varphi \rrbracket(\rho[x  \mapsto v]))^{\alpha} \, (\bot_n)) \, \rho(\vec{\psi}) \\
\llbracket \mathcal{H} \vdash (\nu^\alpha x^T. \varphi) \, \vec{\psi} \rrbracket(\rho) &= ((\lambda v. \llbracket \mathcal{H}, x: T \vdash \varphi \rrbracket(\rho[x  \mapsto v]))^{\alpha} \, (\top_n)) \, \rho(\vec{\psi})
\end{align*}
\fi

\begin{lemma}
Let $\Gamma \vdash \Delta$ be a sequent and $x^\mathbf{N} \in FV(\Gamma, \Delta)$.
If $(\Gamma[\mathbf{S}^n\mathbf{Z}/x] \vdash \Delta[\mathbf{S}^n\mathbf{Z}/x])$ is cut-free provable for all $n \in \mathbb{N}$
then $\Gamma \vdash \Delta$ is cut-free provable.
\end{lemma}
\begin{proof}
Assume $(\Gamma[\mathbf{S}^n\mathbf{Z}/x] \vdash \Delta[\mathbf{S}^n\mathbf{Z}/x])$ is cut-free provable for all $n \in \mathbb{N}$,
and let $\Pi_n$ be the proof.
We define $(\Pi^i)_{i \geq 0}$ recursively whose root node is $(N \, x, \Gamma[\mathbf{S}^ix/x] \vdash \Delta[\mathbf{S}^ix/x])$ as follows:
{\footnotesize
  \begin{center}
  $\Pi^i :=$
  \AxiomC{$\Pi_i$}
  \UnaryInfC{$\Gamma[\mathbf{S}^i\mathbf{Z}/x] \vdash \Delta[\mathbf{S}^i\mathbf{Z}/x]$}
  \RightLabel{(=L)}
  \UnaryInfC{$x=\mathbf{Z}, \Gamma[\mathbf{S}^ix/x] \vdash \Delta[\mathbf{S}^ix/x]$}
    \AxiomC{$\Pi^{i+1}$}
    \UnaryInfC{$N \, x, \Gamma[\mathbf{S}^{i+1}x/x] \vdash \Delta[\mathbf{S}^{i+1}x/x]$}
    \RightLabel{(Subst)}
    \UnaryInfC{$N \, x', \Gamma[\mathbf{S}^i\mathbf{S}x'/x] \vdash \Delta[\mathbf{S}^i\mathbf{S}x'/x]$}
    \RightLabel{($=$L)}
    \UnaryInfC{$x=\mathbf{S}x', N \, x', \Gamma[\mathbf{S}^ix/x] \vdash \Delta[\mathbf{S}^ix/x]$}
    \RightLabel{($\exists$L, $\land$L)}
    \UnaryInfC{$\exists x'. x=\mathbf{S}x' \land N \, x', \Gamma[\mathbf{S}^ix/x] \vdash \Delta[\mathbf{S}^ix/x]$}
    \RightLabel{($\mu$L)}
  \BinaryInfC{$N \, x, \Gamma[\mathbf{S}^ix/x] \vdash \Delta[\mathbf{S}^ix/x]$}
  \DisplayProof
  \end{center}
}
Then
$\Pi :=$
\AxiomC{$\Pi^0$}
\UnaryInfC{$N \, x, \Gamma \vdash \Delta$}
\RightLabel{(Nat)}
\UnaryInfC{$\Gamma \vdash \Delta$}
\DisplayProof
is a pre-proof of $\Gamma \vdash \Delta$.
For all infinite paths $\pi$ in the pre-proof $\Pi$,
if $\pi$ goes through some $\Pi_i$ then there exists left $\mu$-trace or right $\nu$-trace because $\Pi_i$ is a proof,
and otherwise, we can trace the left $\mu$ in $N$.
Therefore, $\Pi$ is a proof of $\Gamma \vdash \Delta$.
\end{proof}

\begin{corollary} \label{co: nat}
Let $\Gamma \vdash \Delta$ be a sequent and $\vec{x}$ are all natural number free variables in $\Gamma \cup \Delta$.
If $(\Gamma[\vec{n}/\vec{x}] \vdash \Delta[\vec{n}/\vec{x}])$ is cut-free provable for all $\vec{n} \in \vec{\mathbb{N}}$
then $\Gamma \vdash \Delta$ is cut-free provable.
\end{corollary}

Thanks to this corollary, it suffices to prove completeness of valid sequents whose free variables have type \( \mathbf{N}^{k} \to \mathbf{\Omega} \) for some \( k \geq 0 \).
In other words, we can assume without loss of generality that the sequent of interest has no free variable of type \( \mathbf{N} \).

The next two definitions construct a tree $T_\omega$ of $\Gamma \vdash \Delta$ without natural number free variables.
{The first definition is needed to deal with all formulas in a fair manner.}

\begin{definition}[Schedule]
  \emph{A schedule element} is a formula of the form $\varphi \lor \psi, \varphi \land \psi, (\lambda x. \varphi) \, \psi \, \vec{\psi}, (\sigma x. \varphi) \, \vec{\psi}$.
  We call $(E_i)_{i \geq 0}$ \emph{a schedule} if $E_i$ is a schedule element for all $i$ and every schedule element appears infinitely often in $(E_i)_{i \geq 0}$.
\end{definition}
There exists a schedule and we fix one.

\begin{definition}[$T_\omega$] \label{def: tomega}
  Let $\Gamma \vdash \Delta$ be a valid sequent that does not have natural number or higher-order free variables.
  That is, the type of each free variable in $\Gamma \vdash \Delta$ is $\mathbf{N}^n \to \mathbf{\Omega}$ for some $n \in \mathbb{N}$.

  Then we will define trees $(T_i)_{i \geq 0}$ whose roots are $\Gamma \vdash \Delta$ inductively by using a schedule $(E_i)_{i \geq 0}$.
  First, $T_0$ is defined by $T_0 := \Gamma \vdash \Delta$.

  Assume $T_i$ is already defined.
  Then we define $T_{i+1}$ in $T_i$ by replacing each open leaf $\Gamma' \vdash \Delta'$ with the following tree:
  \begin{itemize}
    \item If there exists a formula $\varphi \in \Gamma' \cap \Delta'$:

      \begin{center}
        \AxiomC{}
        \RightLabel{(Axiom)}
        \UnaryInfC{$\varphi \vdash \varphi$}
        \RightLabel{(Wk)}
        \UnaryInfC{$\Gamma' \vdash \Delta'$}
        \DisplayProof
      \end{center}

    \item If $(\mathbf{S}^n\mathbf{Z}=\mathbf{S}^m\mathbf{Z}) \in \Gamma'$ for some different natural numbers $n, m$:

      \begin{center}
        \AxiomC{}
        \RightLabel{(P1)}
        \UnaryInfC{$\mathbf{Z} = \mathbf{S}^{|n-m|}\mathbf{Z} \vdash$}
        \RightLabel{(P2)}
        \UnaryInfC{$\mathbf{S}^n\mathbf{Z} = \mathbf{S}^m\mathbf{Z} \vdash$}
        \RightLabel{(Wk)}
        \UnaryInfC{$\Gamma' \vdash \Delta'$}
        \DisplayProof
      \end{center}

    \item If $(\mathbf{S}^n\mathbf{Z}=\mathbf{S}^n\mathbf{Z}) \in \Delta'$ for some $n$:

      \begin{center}
        \AxiomC{}
        \RightLabel{($=$R)}
        \UnaryInfC{$ \vdash \mathbf{S}^n\mathbf{Z} = \mathbf{S}^n\mathbf{Z}$}
        \RightLabel{(Wk)}
        \UnaryInfC{$\Gamma' \vdash \Delta'$}
        \DisplayProof
      \end{center}

    \item Otherwise:
      This replacement is performed in such a way that each formula is {chosen
      as the target of expansion
      in a fair manner, so that every} formula is expanded at some point.
      \begin{itemize}
      \item Case ($E_i \equiv \varphi \lor \psi$):
        \begin{itemize}
          \item if $E_i \in \Gamma'$:
            \begin{center}
            \AxiomC{$\Gamma', \varphi \vdash \Delta'$}
            \AxiomC{$\Gamma', \psi \vdash \Delta'$}
            \RightLabel{($\lor$L)}
            \BinaryInfC{$\Gamma', \varphi \lor \psi \vdash \Delta'$}
            \RightLabel{(Ctr)}
            \UnaryInfC{$\Gamma' \vdash \Delta'$}
            \DisplayProof
            \end{center}

          \item if $E_i \in \Delta'$:
            \begin{center}
              \AxiomC{$\Gamma'\vdash \varphi, \psi, \Delta'$}
              \RightLabel{($\lor$R)}
              \UnaryInfC{$\Gamma' \vdash \varphi \lor \psi, \Delta'$}
              \RightLabel{(Ctr)}
              \UnaryInfC{$\Gamma' \vdash \Delta'$}
              \DisplayProof
            \end{center}
        \end{itemize}
      \item Case $(E_i \equiv \varphi \land \psi)$:
        The tree is defined in the similar way to the case $E_i \equiv \varphi \lor \psi$.

      \item Case $(E_i \equiv (\lambda x. \varphi) \, \psi \, \vec{\psi})$:
        \begin{itemize}
          \item if $E_i \in \Gamma'$:
            \begin{center}
              \AxiomC{$\Gamma', \varphi[\psi/x] \, \vec{\psi} \vdash \Delta'$}
              \RightLabel{($\lambda$L)}
              \UnaryInfC{$\Gamma', (\lambda x. \varphi) \, \psi \, \vec{\psi} \vdash \Delta'$}
              \RightLabel{(Ctr)}
              \UnaryInfC{$\Gamma' \vdash \Delta'$}
              \DisplayProof
            \end{center}

          \item if $E_i \in \Delta'$:
            \begin{center}
              \AxiomC{$\Gamma' \vdash \varphi[\psi/x] \, \vec{\psi}, \Delta'$}
              \RightLabel{($\lambda$R)}
              \UnaryInfC{$\Gamma' \vdash (\lambda x. \varphi) \, \psi \, \vec{\psi}, \Delta'$}
              \RightLabel{(Ctr)}
              \UnaryInfC{$\Gamma' \vdash \Delta'$}
              \DisplayProof
            \end{center}
        \end{itemize}
      \item Case $(E_i \equiv (\sigma x. \varphi) \, \vec{\psi})$:
        \begin{itemize}
          \item if $E_i \in \Gamma'$:
            \begin{center}
              \AxiomC{$\Gamma', \varphi[\sigma x. \varphi/x] \, \vec{\psi} \vdash \Delta'$}
              \RightLabel{($\sigma$L)}
              \UnaryInfC{$\Gamma', (\sigma x. \varphi) \, \vec{\psi} \vdash \Delta'$}
              \RightLabel{(Ctr)}
              \UnaryInfC{$\Gamma' \vdash \Delta'$}
              \DisplayProof
            \end{center}

          \item if $E_i \in \Delta'$:
            \begin{center}
              \AxiomC{$\Gamma' \vdash \varphi[\sigma x. \varphi/x] \, \vec{\psi}, \Delta'$}
              \RightLabel{($\sigma$R)}
              \UnaryInfC{$\Gamma' \vdash (\sigma x. \varphi) \, \vec{\psi}, \Delta'$}
              \RightLabel{(Ctr)}
              \UnaryInfC{$\Gamma' \vdash \Delta'$}
              \DisplayProof
            \end{center}
        \end{itemize}
    \end{itemize}
  \end{itemize}
  Note that for all $i \geq 0$,
  $T_i \subseteq T_{i+1}$
  and each sequent of an open leaf in $T_{i+1}$ includes the sequent of the corresponding leaf in $T_i$.
  We define $T_\omega$ to be $\lim_{i \to \infty} T_i$.
\end{definition}

We aim to prove that $T_\omega$ is a proof of $\Gamma \vdash \Delta$.

\begin{definition}[$\Gamma_\omega \vdash_{n/\pi} \Delta_\omega, \rho_\omega$] \label{def: rhoomega}
  For all open leaves $n$ of $T_\omega$,
  we define $\Gamma_\omega \vdash_n \Delta_\omega$ as the leaf sequent.
  For all infinite paths $\pi = (\pi_i)_{i \geq 0}$ in $T_\omega$,
  we define $\Gamma_\omega \vdash_\pi \Delta_\omega$ as $\lim_{i \to \infty} (\Gamma_i \vdash \Delta_i)$ where $\Gamma_i \vdash \Delta_i$ is the sequent of $\pi_i$.

  A valuation $\rho_\omega$ of $\Gamma_\omega \vdash_{n/\pi} \Delta_\omega$ is defined as below:
  \begin{align*}
    \rho_\omega(x^\mathbf{\Omega})
      &:= \begin{cases} \top & \text{if } x \in \Gamma_\omega \\ \bot & \text{otherwise}\end{cases} \\
    \rho_\omega(f^{\mathbf{N}^n \to \mathbf{\Omega}})
      &:= \lambda \vec{x}^{\mathbf{N}^n}. \begin{cases} \top & \text{if } f \, \vec{x} \in \Gamma_\omega \\ \bot & \text{otherwise}\end{cases}
  \end{align*}
\end{definition}

\begin{lemma} \label{lem: rhoomega}
  $T_\omega$ is a proof.
\end{lemma}
\begin{proof}
  We aim to show that $T_\omega$ is a pre-proof and $T_\omega$ satisfies the global trace condition.

  Assume $T_\omega$ is not a pre-proof.
  There exists an open leaf $n$ in $T_\omega$,
  and all formulas of the leaf are of the form $x^\mathbf{\Omega}$, $f \, \vec{t}$ or $\mathbf{S}^{{m}}\mathbf{Z}=\mathbf{S}^n\mathbf{Z}$.
  Then the definition of $\rho_\omega$ of $\Gamma_\omega \vdash_n \Delta_\omega$ induces
  $\llbracket \varphi \rrbracket_{\rho_\omega} = \top$ for all $\varphi \in \Gamma_\omega$
  and $\llbracket \varphi \rrbracket_{\rho_\omega} = \bot$ for all $\varphi \in \Delta_\omega$.
  This contradicts to $\Gamma \models \Delta$ because $\Gamma \subseteq \Gamma_\omega$ and $\Delta \subseteq \Delta_\omega$.
  Therefore $T_\omega$ is a pre-proof.

  We next show that $T_\omega$ satisfies the global trace condition.
  For all infinite paths $\pi$ in $T_\omega$,
  we define $\rho_\omega$ of $\Gamma_\omega \vdash_\pi \Delta_\omega$ as \autoref{def: rhoomega}.
  We have $\Gamma \models \Delta$, $\Gamma \subseteq \Gamma_\omega$ and $\Delta \subseteq \Delta_\omega$ by the assumption and the construction of $T_\omega$.
  Therefore,
  it follows that there exists
  either (1) a formula $\varphi \in \Gamma_\omega$ such that $\llbracket \varphi \rrbracket_{\rho_\omega} = \bot$ or
  (2) a formula $\varphi \in \Delta_\omega$ such that $\llbracket \varphi \rrbracket_{\rho_\omega} = \top$.
  We now give the proof only for the case (1).
  The other case can be also proved by the same method.

  Let $j$ be a number such that $\varphi \in \pi_j$.
  We will show that there exists a left $\mu$-trace $(\tau_i)_{i \geq j}$ in $\pi$ starting from $\varphi$.

  We define $(\tau_i)_{i \geq j}$ and $(\tau_i')_{i \geq j}$ which satisfy the following conditions:
  \begin{enumerate}
    \item $\tau_i \in \Gamma_i$;
    \item we can get $\tau_i$ by forgetting ordinal numbers of $\tau_i'$;
    \item each $\nu$ in $\tau_i'$ has an ordinal number;
      \iffull
        \item $\llbracket \tau_i' \rrbracket_{\rho_\omega} = \bot$ where $\llbracket . \rrbracket$ for $\nu^\alpha$ is the same as \autoref{ap: soundness};
      \else
        \item $\llbracket \tau_i' \rrbracket_{\rho_\omega} = \bot$;
      \fi
    \item each ordinal number in $\tau_{i+1}'$ is less than or equal to the corresponding ordinal number in $\tau_i'$.
  \end{enumerate}

  The definition of $(\tau_i)_{i \geq j}$ and $(\tau_i')_{i \geq j}$ as below:
  \begin{itemize}
    \item $\tau_j := \varphi$.
      Because $\llbracket \varphi \rrbracket_{\rho_\omega} = \bot$, we can get $\varphi'$ which satisfies $\llbracket \varphi' \rrbracket_{\rho_\omega} = \bot$ by assigning ordinal numbers to all $\nu$ in $\varphi$.
    \item Assume $\tau_i$ and $\tau_i'$ are defined.
      Below, $\psi'$ means the corresponding subformula of $\tau_i'$ for each subformula $\psi$ of $\tau_i$.
      \begin{itemize}
        \item If $\tau_i \equiv x$:

          Since $x \in \Gamma_i$, $\llbracket x \rrbracket_{\rho_\omega} = \top$ because of the definition of $\rho_\omega$.
          By the assumption of $\tau_i'$,
          $\llbracket (x)' \rrbracket_{\rho_\omega} = \llbracket x \rrbracket_{\rho_\omega} = \bot$.
          This is a contradiction.

        \item If $\tau_i \equiv (\mathbf{S}^n\mathbf{Z}=\mathbf{S}^m\mathbf{Z})$:

          By the assumption of $\tau_i'$,
          $\llbracket (\mathbf{S}^n\mathbf{Z}=\mathbf{S}^m\mathbf{Z})' \rrbracket_{\rho_\omega} = \llbracket \mathbf{S}^n\mathbf{Z}=\mathbf{S}^m\mathbf{Z} \rrbracket_{\rho_\omega} = \bot$.
          Then $n \neq m$ holds and $\pi$ is not infinite because of the construction of $T_\omega$.
          This contradicts the fact that $\pi$ is an infinite path.

        \item If $\tau_i \equiv f \, \vec{t}$:

          Since $f \, \vec{t} \in \Gamma_i$, $\llbracket f \, \vec{t} \rrbracket_{\rho_\omega} = \top$ because of the definition of $\rho_\omega$.
          By the assumption of $\tau_i'$,
          $\llbracket (f \, \vec{t})' \rrbracket_{\rho_\omega} = \llbracket f \, \vec{t} \rrbracket_{\rho_\omega} = \bot$.
          This is a contradiction.

        \item If $E_i \not \equiv \tau_i$ then $\tau_{i+1} := \tau_i$ and $\tau_{i+1}' := \tau_i'$.
        \item Otherwise:
          \begin{itemize}
        \item Case $\tau_i \equiv (\psi \lor \chi)$:

          Since $\llbracket \psi' \lor \chi' \rrbracket_{\rho_\omega} = \bot$,
          $\llbracket \psi' \rrbracket_{\rho_\omega} = \llbracket \chi' \rrbracket_{\rho_\omega} = \bot$.
          $\tau_{i+1} := \psi, \tau_{i+1}' := \psi'$ if $\pi$ goes to the left leaf,
          and otherwise $\tau_{i+1} := \chi, \tau_{i+1}' := \chi'$.
          In each case, $\llbracket \tau_{i+1}' \rrbracket_{\rho_\omega} = \bot$ holds.

        \item Case $\tau_i \equiv (\psi \land \chi)$:

          Since $\llbracket \psi' \land \chi' \rrbracket_{\rho_\omega} = \bot$, either $\llbracket \psi' \rrbracket_{\rho_\omega}$ or $\llbracket \chi' \rrbracket_{\rho_\omega}$ is $\bot$.

          $\tau_{i+1} := \psi, \tau_{i+1}' := \psi'$ if $\llbracket \psi' \rrbracket_{\rho_\omega} = \bot$, and otherwise $\tau_{i+1} := \chi, \tau_{i+1}' := \chi'$.

        \item Case $\tau_i \equiv (\lambda x. \psi) \, \chi \, \vec{\theta}$:

          $\tau_{i+1} := \psi[\chi/x] \, \vec{\theta}, \tau_{i+1}' := \psi'[\chi'/x] \, \vec{\theta'}$

          then $\llbracket \psi'[\chi'/x] \, \vec{\theta'} \rrbracket_{\rho_\omega} = \llbracket (\lambda x. \psi') \, \chi' \, \vec{\theta'} \rrbracket_{\rho_\omega} = \bot$.

        \item Case $\tau_i \equiv (\mu x. \psi) \, \vec{\theta}$:

          $\tau_{i+1} := \psi[\mu x. \psi/x] \, \vec{\theta}$, $\tau_{i+1}' := \psi'[\mu x. \psi'/x] \, \vec{\theta'}$

          then $\llbracket (\psi'[\mu x. \psi'/x]) \, \vec{\theta'} \rrbracket_{\rho_\omega} = \llbracket (\mu x. \psi') \, \vec{\theta'} \rrbracket_{\rho_\omega} = \bot$.

        \item Case $\tau_i \equiv (\nu x. \psi) \, \vec{\theta}$:

          Let $\alpha$ be the ordinal number assigned to the head $\nu$ of $(\nu x. \varphi) \, \vec{\psi}$.
          If $\alpha = 0$ then
          $\llbracket (\nu^0 x^T. \varphi') \, \vec{\psi'} \rrbracket_{\rho_\omega}
          = \llbracket (\top_T) \, \vec{\psi'} \rrbracket_{\rho_\omega} = \top$
          and this contradicts the assumption.
          \begin{align*}
          \llbracket (\nu^\alpha x. \psi') \, \vec{\theta'} \rrbracket_{\rho_\omega}
          &= ((\lambda  v. \llbracket \psi' \rrbracket_{\rho_\omega[x \mapsto v]})^{\alpha} \, (\top_n)) \, \llbracket \vec{\theta'} \rrbracket_{\rho_\omega} \\
          &= (\bigsqcap_{\beta < \alpha} (\lambda v. \llbracket \psi' \rrbracket_{\rho_\omega[x \mapsto v]}) \, ((\lambda v. \llbracket \psi' \rrbracket_{\rho_\omega[x \mapsto v]})^{\beta} \, (\top_n))) \, \llbracket \vec{\theta'} \rrbracket_{\rho_\omega} \\
          &= \bot
          \end{align*}

          There exists $\beta < \alpha$ such that
          $((\lambda v. \llbracket \psi' \rrbracket_{\rho_\omega[x \mapsto v]}) \, ((\lambda v. \llbracket \psi' \rrbracket_{\rho_\omega[x \mapsto v]})^{\beta} \, (\top_n))) \, \llbracket \vec{\theta'} \rrbracket_{\rho_\omega} = \bot$.
          $\tau_{i+1} := \psi[\nu x. \psi/x] \, \vec{\theta}$, $\tau_{i+1}' := \psi'[\nu^\beta x. \psi'/x] \, \vec{\theta'}$ then
          \begin{align*}
          \llbracket (\psi'[\nu^\beta x. \psi'/x]) \, \vec{\theta'} \rrbracket_{\rho_\omega}
          &= ((\lambda  v. \llbracket \psi' \rrbracket_{\rho_\omega[x \mapsto v]}) \, (\llbracket \nu^\beta x. \psi' \rrbracket_{\rho_\omega})) \, \llbracket \vec{\theta'} \rrbracket_{\rho_\omega} \\
          &= ((\lambda  v. \llbracket \psi' \rrbracket_{\rho_\omega[x \mapsto v]}) \, ((\lambda v. \llbracket \psi' \rrbracket_{\rho_\omega[x \mapsto v]})^{\beta} \, (\top_n))) \, \llbracket \vec{\theta'} \rrbracket_{\rho_\omega} \\
          &= \bot
          \end{align*}
        \end{itemize}
      \end{itemize}
  \end{itemize}
  If this $(\tau_i)_{i \geq j}$ is a $\nu$-trace,
  there is an infinite sequence $p$ of the trace by the definition of $\nu$-trace.
  For all $n \geq 1$, there is a $\nu_{p[0:n]}$ in some $\tau_i$ and we denote by $\alpha_n$ the ordinal number assigned to this $\nu$ in $\tau_i'$.
  By the definition of $(\tau_i')_{i \geq j}$,
  all ordinal numbers assigned to $\nu_{p[0:n]}$ is equal to $\alpha_n$.
  Then for all $n \geq 1$,
  there exists $i$ such that $(\tau_i \equiv (\nu_{p[0:n]} x. \psi) \, \vec{\theta}) \to (\tau_{i+1} \equiv (\psi[\nu_{p[0:n+1]} x. \psi/x]) \, \vec{\theta})$.
  Hence $\alpha_n > \alpha_{n+1}$ holds for all $n \geq 1$ so $(\alpha_n)_{n \geq 1}$ is a decreasing sequence of ordinal numbers.
  However, such a sequence does not exist, hence a contradiction.
  By \autoref{lem: uniqueness}, it follows that $(\tau_i)_{i \geq j}$ is a left $\mu$-trace of $\pi$.

  Therefore $T_\omega$ satisfies the global trace condition.
\end{proof}

\begin{proof}[Proof (\autoref{completeness}).]
Let $\vec{x}$ be all natural number free variables of $\Gamma \vdash \Delta$.
{By} \autoref{co: nat},
it suffices to show that $\Gamma[\vec{n}/\vec{x}] \vdash \Delta[\vec{n}/\vec{x}]$ is cut-free provable for all $\vec{n} \in \vec{\mathbb{N}}$.
For each $\vec{n} \in \vec{\mathbb{N}}$,
  we construct $T_\omega$ for $\Gamma[\vec{n}/\vec{x}] \vdash \Delta[\vec{n}/\vec{x}]$ as \autoref{def: tomega},
  then \autoref{lem: rhoomega} concludes that $T_\omega$ is a proof.
  Besides, $T_\omega$ is cut-free {by the construction of $T_\omega$.}
\end{proof}

\end{document}